%% file: main.tex
\RequirePackage{amsmath}
\RequirePackage{amsthm}
\RequirePackage{amssymb}
\RequirePackage{manyfoot}

\let\LaTeXcline\cline
\documentclass[iicol,pdflatex]{sn-jnl}
\let\cline\LaTeXcline

\usepackage[utf8]{inputenc}

\usepackage{accents}
\usepackage{balance}
\usepackage{makecell}
\usepackage{caption}
\usepackage{subcaption}

\usepackage[skins]{tcolorbox}

\usepackage{algorithm}
\usepackage[noend]{algpseudocode}
\usepackage[cppcommentstyle,continuouslinenumbers]{smartalgorithmic}

\usepackage[apply,nonotes]{xnotes}
\AddXNotesUser{at}{Andrei}{blue}
\AddXNotesUser{lf}{Luciano}{crimson}

\newtheorem{theorem}{Theorem}[section]

\newtheorem{corollary}[theorem]{Corollary}
\newtheorem{definition}[theorem]{Definition}

\usepackage{xcolor}
\usepackage{amsfonts}
\usepackage{multirow}
\usepackage{relsize}
\usepackage{cleveref}

\input{definitions}

\usepackage{enumitem}

\setlist[enumerate]{itemindent=\dimexpr\labelwidth+\labelsep\relax,leftmargin=0pt}
\setlist[itemize]{itemindent=\dimexpr\labelwidth+\labelsep\relax,leftmargin=0pt}

\AtBeginDocument{%
  }

\begin{document}

\title{Swiper: a new paradigm for efficient weighted distributed protocols}

\author[1]{\fnm{Andrei} \sur{Tonkikh}}\email{tonkikh@telecom-paris.fr}

\author[1]{\fnm{Luciano} \sur{Freitas}}\email{lfreitas@telecom-paris.fr}

\affil[1]{\orgdiv{LTCI}, \orgname{T\'el\'ecom Paris, Institut Polytechnique de Paris}, \orgaddress{\street{19 Place Marguerite Perey}, \city{Palaiseau}, \postcode{91120}, \state{Essonne}, \country{France}}}

\abstract{
The majority of fault-tolerant distributed algorithms are designed assuming a \emph{nominal} corruption model, 
in which at most a fraction $f_n$ of parties can be corrupted by the adversary.
However, due to the infamous Sybil attack, nominal models are not sufficient to express the trust assumptions in open (i.e., permissionless) settings.
Instead, permissionless systems typically operate in a \emph{weighted} model, 
where each participant is associated with a \emph{weight} and the adversary can corrupt a set of parties 
holding at most a fraction $f_w$ of the total weight.

In this paper, we suggest a simple way to transform
a large class of protocols designed for the nominal model into the weighted model. 
To this end, we formalize and solve three novel optimization problems, which we collectively call \emph{the weight reduction problems}, that allow us to map large real weights into small integer weights while preserving the properties necessary for the correctness of the protocols.
In all cases, we manage to keep the sum of the integer weights to be at most linear in the number of parties,
resulting in extremely efficient protocols for the weighted model.
Moreover, we demonstrate that, on weight distributions that emerge in practice, the sum of the integer weights tends to be far from the theoretical worst case and, sometimes, even smaller than the number of participants.

While, for some protocols, our transformation requires an arbitrarily small reduction in resilience (i.e., $f_w = f_n - \epsilon$), surprisingly, for many important problems, we manage to obtain weighted solutions with the same resilience ($f_w = f_n$) as nominal ones.
Notable examples include erasure-coded distributed storage and broadcast protocols, verifiable secret sharing, and asynchronous consensus.
Although there are ad-hoc weighted solutions to some of these problems, the protocols yielded by our transformations enjoy all the benefits of nominal solutions, including simplicity, efficiency, and a wider range of possible cryptographic assumptions.
}


\maketitle

\section{Introduction}

\subsection{Weighted distributed problems}

Traditionally, distributed problems are studied in the egalitarian setting where $n$ parties communicate over a network and any $t$ of them can be faulty or corrupted by a malicious adversary.
Different combinations of $n$ and $t$ are possible depending on the problem at hand, the types of failures (crash, omission, semi-honest, or malicious, also known as Byzantine), and the network model (typically, asynchronous, semi-synchronous, or synchronous).
However, for most distributed protocols, $t$ has to be smaller than a certain fraction of $n$. For example, most practical Byzantine fault-tolerant consensus protocols~\cite{pbft,canetti-rabin} can operate for any $t < \frac{n}{3}$.
We call such models \emph{nominal} and use $\fAn$ to denote their \emph{resilience}, i.e., a nominal protocol with resilience $\fAn$ operates correctly as long as less than $\fAn n$ parties are corrupt, where $n$ is the total number of participants.

However, this simple corruption model is not always sufficient to express the actual fault structure or trust assumptions of real systems.
As a result, we see many practical blockchain protocols adopt a more general, \emph{weighted} model, where each party is associated with a real \emph{weight} that, intuitively, represents the number of ``votes'' this party has in the system.
The assumption on the \emph{number} of corrupt parties in this setting is replaced by the assumption that the \emph{total weight} of the corrupt parties is smaller than a fraction $\fAw$ of the total weight of all participants. 
For example, in permissionless systems, the weight can correspond to the amount of ``stake'' or computational resources a participant has invested in the system and, in the context of managed systems, to a function of the estimated failure probability.

There are two main reasons for adopting the weighted model in the context of blockchain systems.
First and foremost, it protects the system from the infamous \emph{Sybil attacks}, i.e., malicious users registering themselves multiple times in order to obtain multiple identities, thereby surpassing the resilience threshold $\fAn$.
Second, it is speculated that users with a greater amount of resources (monetary, computational, or otherwise) invested in the system, and consequently a higher weight, will be more committed to the system's stability and less likely to engage in malicious behavior.

\subsection{Weighted voting and where it needs help} \label{subsec:weighted-voting}

Perhaps, the most prevalent tool used for the design of distributed protocols is \emph{quorum systems}~\cite{weighted-voting,quorum-systems,byz-quorum-systems}.
Intuitively, to achieve fault tolerance, each ``action'' is confirmed by a sufficiently large set of participants (called a \emph{quorum}).
Then, if two actions are conflicting or somehow interdependent (e.g., writing and reading a file in a distributed storage system), then the parties in the intersection of the quorums are supposed to ensure consistency.
Thus, many distributed protocols can be converted from the nominal to the weighted setting simply by changing the quorum system, i.e., instead of waiting for confirmations from a certain number of parties, waiting for a set of parties with the corresponding fraction of the total weight.
We call this strategy \emph{weighted voting} and it often allows translating protocols from the nominal to the weighted model while maintaining the same resilience (i.e., $\fAw = \fAn$) and, in some cases, with virtually no overhead.

However, weighted voting has two major downsides.
First and foremost, many protocols rely on primitives beyond simple quorum systems, and weighted voting is often insufficient to translate these protocols to the weighted model.
Notable examples include threshold cryptography~\cite{threshold-cryptosystems,ThresholdBLS}, secret sharing~\cite{Sha79,avss-cachin-2002}, erasure and error-correcting codes~\cite{coding-theory-textbook}, and numerous protocols that rely on these primitives. %

Another example, relevant to blockchain systems,\atadd{ where weighted voting is typically not sufficient} is\atadd{ in} Single Secret Leader Election protocols~\cite{boneh2020single,adaptive-ssle,uc-ssle, freitas2023homomorphic}.
It illustrates that not all protocols that \atrev{cannot} be\atremove{ easily} converted to the weighted model\atadd{ simply3} by applying weighted voting belong to the categories above and motivates the general approach taken in this paper.

The second drawback of weighted voting is that it requires a careful examination of the protocol in order to determine whether weighted voting is sufficient to convert it to the weighted model, as well as non-trivial modifications to the protocol implementation. 
It would be much nicer to have a ``black-box'' transformation that would take a protocol designed and implemented for the nominal model and output a protocol for the weighted model.
\atrev{In this paper, we offer both a ``black-box'' transformation and a set of more efficient ``open-box'' transformations for a wide range of problems.}

\subsection{Our contribution} \label{subsec:contribution}

Our contribution to the fields of distributed computing and applied cryptography is twofold:
\begin{enumerate}
    \item We present a simple and efficient black-box transformation that can be applied to convert a wide range of protocols designed for the nominal model into the weighted model.
    Crucially, one can determine the applicability of the transformation simply by examining the \emph{problem} in question (e.g., Byzantine consensus), instead of the \emph{protocol} itself (e.g., PBFT~\cite{pbft}) and it does not require modifications to the source code, only a slim wrapper around it.
    The price for this transformation is an arbitrarily small decrease in resilience ($\fAw = \fAn - \epsilon$, where $\epsilon > 0$) and an increase in the communication and computation complexities proportional to $\frac{\fAw}{\epsilon}$.

    \item Furthermore, by opening the black box and examining the internal structure of distributed protocols, we discover that by combining our transformation with weighted voting, in many cases, we can obtain weighted algorithms \emph{without} the reduction in resilience ($\fAw = \fAn$) and with a minor or non-existent performance penalty.
\end{enumerate}

\input{table-applications}

We summarize some examples of our techniques applied to a range of different protocols in \Cref{tab:applications}. The last two columns of the table give the upper bound on the overhead of the obtained weighted protocols compared to their nominal counterparts executed with the same number of parties.
Note, however, that, in many cases, the overhead applies only to specific parts of the protocol, which may not be the bottlenecks.
Thus, further experimental studies may reveal that the real overhead is even lower or nonexistent, even with the worst-case weight distribution.
Columns ``$\fAw$'' and ``$\fAn$'' specify the resilience of the weighted protocols obtained and the original nominal protocols, respectively. As discussed above, in most cases\atadd{,} we manage to avoid sacrificing resilience (i.e., $\fAw = \fAn$).

Furthermore, the main building block of our constructions, the \emph{weight reduction problems}, may be of separate interest and may have important applications beyond distributed protocols.
It is indeed an interesting and somewhat counterintuitive observation that large real weights can be efficiently reduced to small integer weights while preserving the key structural properties. 
We formally define the three weight reduction problems considered in this paper in \Cref{sec:problem-statement} and present a practical solver called {\Swiper} in \Cref{sec:algo}.

\subsection{Empirical study} \label{subsec:intro:empirical-findings}

The performance of the weighted protocols constructed as suggested in this paper is sensitive to the distribution of\atadd{ the participants'} weights\atremove{ of the participants}.
While we provide upper bounds and thus analyze our protocols for the worst weight distributions possible, it is interesting whether such weight distributions emerge in practice.

To study real-world weight distributions, we tested our weight reduction algorithms on the distribution of funds from multiple existing blockchain systems~\cite{aptos_white,tezos_white,filecoin_white,algorand_white} ranging in size from a hundred parties~\cite{aptos_stake,aptos_white} up to multiple tens of thousands~\cite{algorand_stake,algorand_white}.
We perform an in-depth analysis in \Cref{sec:empirical-study}.

\subsection*{Roadmap}

The rest of the paper is organized as follows: we formally define weight reduction problems and state the upper bounds in \Cref{sec:problem-statement}.
We present {\Swiper} in \Cref{sec:algo}.
The proof that it satisfies the stated bounds is delegated to \Cref{sec:proof}.
\Cref{sec:wr-applications,sec:wq-applications,sec:derived-applications} discuss in detail the applications of the weight reduction problems in distributed computing and cryptography.
In \Cref{sec:empirical-study}, we discuss the results of the empirical study performed on real-world weight distributions.
We discuss related work in \Cref{sec:related} and conclude the paper with the discussion of directions for future work in \Cref{sec:conclusion}.

To avoid distraction, we moved to the appendix the parts of the paper that are, while essential, not required for understanding the key ideas.
Specifically, the formal proofs of the upper bounds (\Cref{sec:proof}), the mixed integer programming formulation of the {\WRFull} problem (\Cref{sec:mip}), and the empirical evaluation results (\Cref{sec:remainingplots}).

\section{Weight reduction problems} \label{sec:problem-statement}

Let us define the key building block to our construction, the \emph{weight reduction problems}: a class of optimization problems that map (potentially large) real weights $w_1, \dots, w_n \in \mathbb{R}_{\ge 0}$ to (ideally small) integer weights $t_1, \dots, t_n \in \mathbb{Z}_{\ge 0}$ while preserving certain key properties.
For convenience, we use the word \emph{``tickets''} to denote the units of the assigned integer weights, i.e., if $t_1, \dots, t_n$ is the output of a weight reduction problem, we say that party $i$ is given $t_i$ \emph{tickets}.

\begin{notation}
    To avoid repetition, throughout the rest of the paper, we use the following notation:
    \begin{enumerate}
        \item $[n] := \{1, 2, \dots, n\}$
        \item for any $S \subseteq [n]$: $w(S) := \sum_{i \in S} w_i$
        \item for any $S \subseteq [n]$: $t(S) := \sum_{i \in S} t_i$
                \item $W := w([n]) = \sum_{i=1}^{n} w_i$
        \item $T := t([n]) = \sum_{i=1}^{n} t_i$
    \end{enumerate}
\end{notation}

\subsection{\WRFull}

The first weight reduction problem is \emph{\WRFull} (or simply {\WR}).
It is parameterized by two numbers $\fRw, \fRn \in (0, 1)$ and requires the mapping to preserve the property that any subset of parties of weight less than $\fRw$ obtains less than $\fRn$ tickets.
More formally:
\begin{statement}[\WRFull]
    \begin{atreview}      Given $\fRw, \fRn \in (0, 1)$ and $w_1, \dots, w_n \in \mathbb{R}_{\ge 0}$\atadd{ such that $W \neq 0$} as input, find $t_1, \dots, t_n \in \mathbb{Z}_{\ge 0}$ such that \atreplace{$\sum_{i=1}^n t_i$}{$T$} is minimized, subject to the following restriction:
    
    \begin{center}
        $\forall S \subseteq [n]\text{ s.t. }w(S) < \fRw W: t(s) < \fRn T$
    \end{center}
    \end{atreview}
    \label{def:wrp}
\end{statement}

In \Cref{sec:wr-applications}, we apply {\WRFull} to implement the black-box transformation announced in \Cref{subsec:contribution} as well as weighted versions of secret sharing and threshold cryptography with different access structures.
In~\Cref{sec:proof}, we prove the following theorem:
\begin{theorem}[{\WR} upper bound] \label{thm:wr-bound}
    For any $\fRw, \fRn \in (0, 1)$ such that $\fRw < \fRn$\atadd{ and any $w_1, \dots, w_n$}: there exists a solution to the {\WRFull} problem with
        $T \le \left\lceil \frac{\fRw(1 - \fRw)}{\fRn - \fRw} n \right\rceil$
\end{theorem}

To make sense of this expression, note that: (1) it is proportional to $n$; (2) it is inversely proportional to the ``gap'' between $\fRw$ and $\fRn$; (3) the numerator $\fRw (1 - \fRw)$ is smaller than $1$ and, in fact, never exceeds $1/4$. For a fixed $\fRw$, one can see $\fRw (1 - \fRw)$ as the ``constant'' and $O\left(\frac{n}{\fRn - \fRw}\right)$ as the ``complexity''.

\subsection{\WQFull}

The next weight reduction problem we study is \emph{\WQFull} (or simply {\WQ}).
It requires the mapping to preserve the property that any subset of parties of weight greater than $\fQw$ obtains more than $\fQn$ tickets.
In some sense, {\WQ} is the opposite of the {\WRFull} problem discussed above.
More formally:

\begin{statement}[\WQFull]
    \begin{atreview}      Given $\fQw, \fQn \in (0, 1)$ and $w_1, \dots, w_n \in \mathbb{R}_{\ge 0}$\atadd{ such that $W \neq 0$} as input, find $t_1, \dots, t_n \in \mathbb{Z}_{\ge 0}$ such that \atreplace{$\sum_{i=1}^n t_i$}{$T$} is minimized, subject to the following restriction:
    
    \begin{center}
        $\forall S \subseteq [n]\text{ s.t. }w(S) > \fQw W: t(s) > \fQn T$
            \end{center}
    \end{atreview}
    \label{def:wqp}
\end{statement}

In \Cref{sec:wq-applications}, we show how to apply {\WQFull} to implement weighted versions of storage and broadcast protocols that rely on erasure and error-correcting codes for minimizing communication and storage complexity.

There exists a simple reduction between {\WR} and {\WQ}:

\begin{theorem} \label{thm:wr2wq}
    For any $\fQw, \fQn \in (0, 1)$ and $w_1, \dots, w_n \in \mathbb{R}_{\ge 0}$, the following problems are identical:
    \begin{enumerate}
        \item $\WQ(\fQw, \fQn, w_1, \dots, w_n)$
        \item $\WR(1-\fQw, 1-\fQn, w_1, \dots, w_n)$
    \end{enumerate}
\end{theorem}
\begin{proof}
    Let us prove that any valid solution to $\WR(1-\fQw, 1-\fQn, w_1, \dots, w_n)$ is a valid solution to $\WQ(\fQw, \fQn, w_1, \dots, w_n)$. 
    The inverse can be proven analogously.
        Indeed, if\atadd{ $t_1, \dots, t_n$ is a valid solution for $\WR(1-\fQw, 1-\fQn, w_1, \dots, w_n)$, then} $\forall S \subseteq [n]$ such that $w(S) > \fQw W$\atreplace{:}{ it holds that} 
    $w([n] \setminus S) = W - w(S) < (1 - \fQw) W$.
    Hence, $t([n] \setminus S) < (1 - \fQn) T$
    and $t(S) = T - t([n] \setminus S) > \fQn T$.
\end{proof}

From \Cref{thm:wr-bound,thm:wr2wq}, we obtain the following:
\begin{corollary}[{\WQ} upper bound] \label{col:wq-bound}
    For any $\fQw, \fQn \in (0, 1)$ such that $\fQn < \fQw$: there exists a solution to the {\WQFull} problem with $T \le \left\lceil \frac{\fQw(1 - \fQw)}{\fQw - \fQn} n\right\rceil$
        \end{corollary}

\subsection{\WSFull}

Finally, {\WSFull}, in a sense, combines {\WR} and {\WQ}: it has\atadd{ two} parameters\atadd{,} $\alpha$ and $\beta$\atadd{,} and guarantees that any set of weight $\beta$ receives more tickets than any set of weight $\alpha$.
Intuitively, it is similar to solving ${\WR}(\alpha, \gamma)$ and $\WQ(\beta, \gamma)$ for some unknown \atreplace{$\gamma \in (\alpha, \beta)$}{$\gamma \in (0, 1)$} \emph{at the same time}, i.e., with just a single ticket assignment.

\begin{statement}[\WSFull]
    \begin{atreview}
    Given $\alpha, \beta \in (0, 1)$ and $w_1, \dots, w_n \in \mathbb{R}_{\ge 0}$\atadd{ such that $W \neq 0$} as input, find $t_1, \dots, t_n \in \mathbb{Z}_{\ge 0}$ such that \atreplace{$\sum_{i=1}^n t_i$}{$T$} is minimized, subject to the following restriction:
    
    \begin{center}
        $\forall S_1, S_2 \subseteq [n]\text{ s.t. }w(S_1) < \alpha W$ and $w(S_2) > \beta W$:\\
         $t(S_1) < t(S_2)$
    \end{center}
    \end{atreview}
    \label{def:wsp}
\end{statement}

In this paper, we focus primarily on {\WRFull} and {\WQFull} because they are sufficient for most applications and, being less restrictive on the ticket assignment, permit more efficient solutions. However, for completeness, we also provide an upper bound on {\WSFull}\atadd{ and support it in our approximate solver described in \Cref{sec:algo}}.

\begin{theorem}[{\WS} upper bound] \label{thm:ws-bound}
For any $\alpha, \beta \in (0, 1)$ such that $\alpha < \beta$: there exists a solution to the {\WSFull} problem with $T \le \frac{(\alpha + \beta)(1 - \alpha)}{\beta - \alpha} n$.
\end{theorem}

Note that the numerator $(\alpha + \beta)(1 - \alpha)$ is always smaller than 1 for $0 < \alpha < \beta < 1$.

\section{Swiper: Approximate solver for Weight Reduction problems}\label{sec:algo}

\begin{table*}[htb!]
    \centering
    \small
    \newcommand{\tp}[1]{{\smaller #1}}
    \resizebox{\linewidth}{!}{
    \begin{tabular}{|c|c|c|c|c|c|c|c|}
    \hline
        \multirow{6}{*}{\thead{System}}
        & \multicolumn{7}{c|}{\thead{number of tickets using {\Swiper}}}
    \\ \cline{2-8}
        & \multicolumn{4}{c|}{\thead{{\WR} and {\WQ}}}
        & \multicolumn{3}{c|}{\thead{{\WS}}}
    \\ \cline{2-8}
                & \makecell[c]{\tp{$\fRw = 1/4$} \\ \tp{$\fRn = 1/3$}} 
        & \makecell[c]{\tp{$\fRw = 1/3$} \\ \tp{$\fRn = 3/8$}} 
        & \makecell[c]{\tp{$\fRw = 1/3$} \\ \tp{$\fRn = 1/2$}} 
        & \makecell[c]{\tp{$\fRw = 2/3$} \\ \tp{$\fRn = 3/4$}} 
        & \multirow{3}{*}{\makecell[c]{\tp{$\alpha = 1/4$} \\ \tp{$\beta = 1/3$}}}
        & \multirow{3}{*}{\makecell[c]{\tp{$\alpha = 1/3$} \\ \tp{$\beta = 1/2$}}}
        & \multirow{3}{*}{\makecell[c]{\tp{$\alpha = 2/3$} \\ \tp{$\beta = 3/4$}}}
        \\
    \cline{2-5}
                & \makecell[c]{\tp{$\fQw = 3/4$} \\ \tp{$\fQn = 2/3$}} 
        & \makecell[c]{\tp{$\fQw = 2/3$} \\ \tp{$\fQn = 5/8$}} 
        & \makecell[c]{\tp{$\fQw = 2/3$} \\ \tp{$\fQn = 1/2$}} 
        & \makecell[c]{\tp{$\fQw = 1/3$} \\ \tp{$\fQn = 1/4$}}
        & %
        & %
        & %
                \\
    \hline\hline
    \makecell[c]{Aptos~\cite{aptos_white,aptos_stake} \\ $W = 8.47\times10^8$ \; $n = 104$}
        & 85  & 235 & 27 & 110 & 385 & 98 & 437 (+1) \\
    \hline
    \makecell[c]{Tezos~\cite{tezos_white, tezos_stake} \\ $W = 6.76\times10^8$ \; $n = 382$}
        & 133 & 425 & 61 (+8) & 258 (+1) & 670 & 233 (+2) & 811 \\
    \hline
    \makecell[c]{Filecoin~\cite{filecoin_white,filecoin_stake} \\ $W = 2.52\times10^{19}$ \; $n = 3\,700$}
        & 3\,091 & 8\,233 & 1\,533 & 4\,691 & 10\,485 & 4\,838 & 11\,858 \\
    \hline
    \makecell[c]{Algorand~\cite{algorand_white,algorand_stake} \\ $W = 9.72\times10^9$ \; $n = 42\,920$}
        & 745 & 13\,475 & 293 & 6\,258 & 46\,009 & 2\,188 & 64\,189 \\
    \hline
    \end{tabular}
    }
    
    \medskip
    \caption{Number of tickets allocated by the {\Swiper} protocol on sample weight distributions.\atremove{ W denotes the total weight of the parties in the blockchain, while N denotes their total number.}\atadd{ In the few cases when the linear mode yields more tickets than the standard (full) mode, the difference is written in parentheses.}}
    \label{tab:evaluations}
\end{table*}

\sloppy
To provide a constructive proof for \Cref{thm:wr-bound,col:wq-bound,thm:ws-bound} as well as to facilitate practical applications of weight reduction problems, we designed {\Swiper}\atrev{---}a fast approximate solver for the three weight reduction problems defined in this paper.
{\Swiper} enjoys a number of desirable properties:
\begin{enumerate}
    \item \textbf{Robustness:} It always respects the upper bounds stated in \Cref{sec:problem-statement}. This means that even under a malicious distribution of weights, the number of assigned tickets will be within a known limit, linear in the number of parties.
    We present the algorithm in \Cref{subsec:algorithm} and prove the upper bounds in \Cref{sec:proof}.

    \item \textbf{Determinism:} {\Swiper} is a deterministic protocol. Hence, when the initial weights are common knowledge, each party can run it locally and all parties will obtain the same result. This eliminates the need for executing \atreplace{a complex consensus protocol}{any complex protocol} to agree on the ticket assignment.

    \item \textbf{Allocation efficiency:} 
        \atrev{As we explore in detail in Section 7, Swiper performs remarkably well on real-world weight distributions, often allocating far fewer tickets than predicted by the upper bounds.}
    In \Cref{tab:evaluations}, we summarize the number of tickets allocated by {\Swiper} on the distribution of funds in four major blockchain systems~\cite{aptos_stake,tezos_stake,filecoin_stake,algorand_stake} with some example thresholds. Notice that, in many cases, the number of tickets is actually below the number of users. This happens partly due to the distributions being significantly skewed and a large number of users actually owning only a small fraction of the total funds.

        \item \textbf{Computational efficiency:} \atrev{Assuming that the thresholds ($ \alpha, \fRw, \fRn, \beta, \fQw, \fQn$) are constants, the runtime of {\Swiper} is either $\tilde{O}(n)$ (in \mbox{\normalfont\texttt{{-}{-}linear}} mode) or $\tilde{O}(n^2)$ (in standard mode).
    The difference in the implementation of the two modes is detailed in \Cref{subsec:algorithm}.}
    Both modes respect the upper bounds and, as can be seen in \Cref{tab:evaluations}, in practice, usually yield identical or almost identical results.
                    \end{enumerate}

\subsection{Algorithm and implementation} \label{subsec:algorithm}

\myparagraph{Overall structure}
In {\Swiper}, we consider ticket assignments of a special form.
Let $c$ be a fixed number between $0$ and $1$ (we will precisely specify $c$ later in this section).
Let $t(s, k)$ be the result of the following procedure: first, let $t_i := \lfloor s w_i + c \rfloor$; then, consider the parties that ended up ``on the border'', i.e., that would lose a ticket if we decreased $s$ any further\footnote{This corresponds to all $i$ such that $s w_i + c$ is an integer.} and take 1 ticket from all but arbitrary (yet deterministically chosen) $k$ of them.

More formally, let $\mathcal{B}_s := \{i \mid s w_i + c\text{ is integer}\}$ and $\mathcal{K}_{s,k} := \{\text{arbitrary $k$ members of $\mathcal{B}_s$}\}$.
Then:
\begin{center}
    $t(s, k)_i := \begin{cases}
        \lfloor s w_i + c \rfloor - 1, &\text{if $i \in (\mathcal{B}_s \setminus \mathcal{K}_{s,k})$} \\
        \lfloor s w_i + c \rfloor, &\text{otherwise}
    \end{cases}$
\end{center}

The crucial observation is that, despite having two indices, this family of ticket assignments can be totally ordered, each ticket assignment having precisely one ticket more than the previous one (after removing duplicates).
Indeed, let $T_{s,k} := \sum_{i=1}^{n} t(s,k)_i$.
Then, for $0 < k < |\mathcal{B}_s|$, by definition, $T(s, k+1) = T(s, k) + 1$. 
Moreover, if $s'$ is the smallest number greater than $s$ such that $|\mathcal{B}_{s'}| \neq 0$, then $T(s', 1) = T(s', 0) + 1 = T(s, |\mathcal{B}_s|) + 1$.
For any $s''$ in between $s$ and $s'$, $t(s'', *) = t(s', 0) = t(s, |\mathcal{B}_s|)$.

{\Swiper} finds a \emph{local minimum} in this family of ticket assignments, i.e., $s^*$ and $k^*$ such that $t(s^*, k^*)$ is viable (satisfies the problem requirements), but, for any sufficiently small $\varepsilon$ and any $k'$, $t(s^* - \varepsilon, k')$ is not viable and neither is $t(s^*, k^* - 1)$.

\myparagraph{Theoretical foundations}
In \Cref{sec:proof}, we prove that, by selecting the constant $c$ as $\fRw$ in case of {\WRFull}, $(1 - \beta_w)$ in case of {\WQFull}, and $\frac{\alpha + \beta}{2}$ in case of {\WSFull}, the resulting ticket assignment always satisfies the bounds stated in \Cref{sec:problem-statement} (\Cref{thm:wr-bound,col:wq-bound,thm:ws-bound}).\footnote{To obtain these specific values, we first considered the general case for an arbitrary $c$ and then found the values of $c$ that minimized the upper bounds.}
The proof works by demonstrating that any invalid ticket assignment of this form yields at least one fewer tickets than the stated upper bounds. Any local minimum yields just 1 ticket more than \emph{some} invalid ticket assignment and, thus, \emph{fewer or equal} to the upper bounds. This proof structure is important for achieving practical efficiency.

\myparagraph{Bootstrapping the solution}
As mentioned above, 
if the ticket assignment yielded by some tuple $(s, k)$ 
is invalid (does not satisfy the problem's requirement), then the total number of tickets in this ticket assignment must be smaller than the upper bound.
Conversely, 
if some tuple $(s, k)$ yields a ticket assignment with the total number of tickets greater or equal to the upper bound,
we can conclude that the resulting ticket assignment is valid.
This fact alone allows us to quickly arrive at a valid solution satisfying the upper bound by simply finding a tuple $(s, k)$ that yields the number of tickets exactly equal to the upper bound.
This can be done efficiently with a binary search.

\myparagraph{Finding the local minimum}

Thanks to the fact that we are only looking for a \emph{local} minimum in the considered family of ticket assignments,
we can find it efficiently with a binary search,
assuming an efficient algorithm for verifying the validity of a ticket assignment.
However, \emph{in the general case}, verifying the validity of a ticket assignment looks a lot like a (co-)NP-hard problem.
Indeed, one can easily see that verifying a solution to {\WRFull} as defined in \Cref{sec:problem-statement} is equivalent to solving a particular instance of Knapsack---the famous NP-hard optimization problem~\cite{KnapsackProblems04}.

Fortunately, for the specific family of ticket assignments that {\Swiper} considers (denoted as $t(s, k)$ earlier in this section), an efficient algorithm does exist.
Indeed, we have already established that any ticket assignment in this family with the total number of tickets ($T$) exceeding the upper bound is valid.
If, on the other hand, $T$ is smaller than the upper bound, then we can use the ``dynamic programming by profits'' approach~\cite[Lemma 2.3.2]{KnapsackProblems04} to solve Knapsack in time $O(Tn)$.
Assuming $\alpha$s and $\beta$s to be constant, $T$ is $O(n)$ and $O(Tn) = O(n^2)$.

\myparagraph{Practical efficiency and the \mbox{\normalfont\texttt{{-}{-}linear}} mode}
Solving Knapsack to verify the validity of $t(s, k)$ is the main bottleneck for the algorithm. To achieve better practical efficiency, {\Swiper} uses well-known quasilinear-time Knapsack lower and upper bounds to filter out as many solutions as possible without invoking the knapsack solver.

The upper bound allows us to implement a \emph{conservative check}, i.e., it may yield false negatives (falsely declaring $t(s,k)$ as invalid), but never false positives (falsely declaring $t(s,k)$ as valid).
In \mbox{\normalfont\texttt{{-}{-}linear}} mode, {\Swiper} only relies on the upper bound and is guaranteed to find a valid solution, albeit not necessarily a locally minimal one.

Additionally, the lower bound allows us to implement a \emph{liberal check}, i.e., it may yield false positives, but never false negatives.
By combining the two, we can implement a quick test that can return one of the three values (``valid'', ``invalid'', or ``uncertain'').
In the full mode (i.e., when \mbox{\normalfont\texttt{{-}{-}linear}} is not provided), {\Swiper} only invokes the full knapsack solver (with $O(n^2)$ time complexity) when the quick test returns ``uncertain'', which speeds up the algorithm by a more than a factor of 3 on inputs with large enough resulting number of tickets.

\myparagraph{Prototype implementation}
We provide the full code for a prototype of {\Swiper} and the data used to generate \Cref{tab:evaluations} in a public GitHub repository\footnote{\url{https://github.com/DCL-TelecomParis/swiper}}.
The prototype is implemented in Python, with JIT compilation used for certain computation-heavy parts.
It utilizes the \texttt{Fraction} class to avoid any possible rounding errors.
If sub-second latencies are required by the application, an implementation in a more performance-oriented programming language as well as the use of rounding (be it floating- or fixed-point arithmetic) may be necessary.

\section{Applications of {\WRFull}} \label{sec:wr-applications}

\subsection{Distributed random number generation}
\label{sec:random-beacon}

As a motivating example for {\WRFull}, consider the \emph{Distributed Random Number Generation} problem. Typically, it needs to satisfy two properties:
\begin{itemize}
    \item If \emph{all} honest parties cooperate, they can generate the next random number;
    
    \item Unless \emph{at least one} honest party wants to open the next random number, it remains completely unpredictable to the adversary.
\end{itemize}

Perhaps, the simplest way it can be achieved~\cite{pre-shared-common-coins} is by having a trusted party generate the random number and pre-distribute it using secret sharing~\cite{Sha79}, such that each party gets a number $t_i$ of shares and any subset of parties possessing at least $\lceil \fRn T\rceil$ shares (where $T = \sum_{i=1}^n t_i$) can reconstruct the secret, but no set of parties possessing less than this amount of shares can learn anything about the secret.

Thus, by setting $\fRw$ to the resilience of the protocol ($\fRw := \fAw$) and $\fRn \le \frac{1}{2}$, we can guarantee that:
\begin{itemize}
    \item Honest participants will receive more than $(1 - \fRn) T \ge \lceil \fRn T\rceil$ shares and, hence, will be able to reconstruct the random number.
    
    \item Corrupt participants will receive less than $\fRn T$ shares and, hence, will not be able to reconstruct the random number unless some honest party also wants to open it;
\end{itemize}

Practical randomness beacons~\cite{random-oracles,sok-randomness-beacons} operate similarly, only employing \emph{unique threshold signatures}~\cite{rsa-threshold-signatures,ThresholdBLS} in order to be able to reuse the same secret multiple times.
The described weighted solution still applies to such approaches unchanged.

\subsection{Blunt Secret Sharing and derivatives} \label{subsec:blunt}

In cryptography, certain actions have an associated access structure $\mathbb{A}$ that determines all sets of parties that are able to perform these actions once they collaborate.
Traditional $(n, k+1)$-threshold systems can be seen as a particular access structure $\mathbb{A}_{n}(\alpha) = \{P \subseteq [n]: |P| > \alpha n\}$, where $\alpha := \frac{k}{n}$.
Analogously, a \emph{weighted} threshold access structure can be defined as
$\mathbb{A}_w(\alpha) = \{P \subseteq \Pi: \sum_{i \in P} w_i > \alpha \sum_{i \in \Pi} w_i\}$.

We can also define the \emph{adversary structure} $\mathbb{F} \subseteq 2^\Pi$, the set of all sets of parties that can be simultaneously corrupted at any given execution. Often, the adversary structure is also defined by a threshold, with a maximum corruptible weight fraction $\fAw$, i.e., $\mathbb{F}_w(\fAw) = \{P \subset \Pi: \sum_{i \in P} w_i < \fAw \sum_{i \in \Pi} w_i\}$.

While threshold access structures are commonly studied in cryptography and are applied in numerous distributed protocols, in practice, as we illustrate in \Cref{sec:derived-applications}, it is often sufficient if the access structure provides the following two properties, generalizing the requirements of the random beacon presented in \Cref{sec:random-beacon}:

\begin{itemize}
    \item There exists at least one set entirely composed of honest parties that belongs to the access structure. This typically guarantees the accompanying protocol's \emph{liveness properties}.
    \item Any set containing only corrupt parties does not belong to the access structure, as this would break \emph{safety properties}. 
\end{itemize}

Hence, we define a \emph{blunt access structure} as follows:

\begin{definition}[Blunt access structure] \label{def:blunt-as}
    Given a set of parties $\Pi$ and the adversary structure $\mathbb{F} \subseteq 2^{\Pi}$, $\mathbb{A}$ is a blunt access structure w.r.t. $\mathbb{F}$ if $(\forall F \in \mathbb{F}: F \not \in \mathbb{A}) \text{ and } (\exists A \in \mathbb{A}: A \cap F = \varnothing)$.
\end{definition}

The following theorem shows that solving {\WR} is sufficient to implement weighted cryptographic protocols with blunt access structures by a reduction to their nominal counterparts.

\begin{theorem}
    Given a set of parties, a protocol $\protocol$ implementing a cryptographic primitive with nominal threshold access structure $\mathbb{A}_n(\fRn)$, for $\fRn \le \frac{1}{2}$, we obtain a protocol $\protocol'$ implementing a blunt access structure w.r.t. adversarial structure $\mathbb{F}_w(\fAw)$, assuming $\fAw < \fRn$, by solving {\WRFull} with the corresponding parameters $\fRn$ and $\fRw := \fAw$.
    This is accomplished by instantiating $\protocol$ with $\hat{n} = T$ virtual users and allowing party $i$ to control $t_i$ of them.%
    \footnote{Recall that $t_i$ is the number of tickets assigned to party $i$ and $T$ is the total number of tickets assigned by the solution to the weight reduction problem (in this case, to {\WR}). See \Cref{sec:problem-statement} for details.}
    \label{thm:reductioncorrect}
\end{theorem}
\begin{proof}
        By definition of {\WR}, once it distributes $T$ tickets, the number of tickets (and, hence, virtual users) allocated to the corrupt parties will be less than $\fRn T$. Hence, no element of the adversary structure shall appear in the resulting access structure. In addition, honest participants will receive more than $(1-\fRn)T \ge \fRn T$ (recall that $\fRn \le \frac{1}{2}$) tickets (and, hence, virtual users), ensuring that there exists a set consisting of only honest parties in the access structure.
\end{proof}

Note that all participants must agree on how many virtual users are assigned to each party, as nominal protocols typically assume that the membership is common knowledge. To this end, it is sufficient for all parties to run an agreed upon \emph{deterministic} weight-restriction protocol (e.g., {\Swiper}).

Among other things, this way, one can obtain weighted versions of secret sharing~\cite{Sha79}, distributed random number generation~\cite{random-oracles}, threshold signatures~\cite{ThresholdBLS}, threshold encryption~\cite{threshold-cryptosystems}, and threshold fully-homomorphic encryption~\cite{ThFHE-jain}, all with blunt access structures.
In the next section, we discuss how to do it for other access structures.

\subsection{Tight Secret Sharing and derivatives} \label{subsec:tight-secret-sharing}

Although a blunt access structure is sufficient for a large spectrum of applications, more restrictive access structures are sometimes necessary as well.
Here, we present a straightforward approach that involves just one extra round of communication to transform a blunt access structure 
into a weighted threshold access structure.\footnote{In fact, this can be further generalized to arbitrary access structures.}
This means that our construction can be readily 
utilized in any protocol that already uses threshold cryptography without requiring significant redesign efforts.

Given a protocol $\protocol$ implementing a certain primitive of distributed cryptography (e.g., threshold signatures~\cite{threshold-cryptosystems}) with a blunt access structure, we can obtain a protocol $\protocol'$ implementing the same protocol with a weighted threshold access structure $\mathbb{A}_w(\beta)$ as follows:
whenever an honest party wants to perform an action $\mathcal{A}$ (e.g., produce a threshold signature), instead it simply broadcasts a message ``voting'' for the action to be performed, without actually revealing any secret data (e.g., its threshold signature share). Then, when an honest party receives such votes from parties with a total weight more than $\beta W$, it participates in the action $\mathcal{A}$, according to the underlying protocol $\protocol$ (e.g., broadcasts its threshold signature share).
Thus, we can notice that:
\begin{enumerate}
    \item Unless a threshold of parties (potentially including Byzantine) cast votes for $\mathcal{A}$, no honest party will participate in $\mathcal{A}$ in $\protocol$. Thus, by \Cref{def:blunt-as}, action $\mathcal{A}$ will not be performed;
    
    \item If a threshold of parties cast votes for $\mathcal{A}$, all honest parties will eventually participate in $\mathcal{A}$ according to $\protocol$, thus, by \Cref{def:blunt-as}, the action will be performed.
\end{enumerate}

\subsection{Black-Box transformation} \label{subsec:black-box}

The same approach of allocating a number of virtual users according to the number of tickets as described in \Cref{subsec:blunt} can be applied to arbitrary distributed protocols.

Given a nominal protocol $\protocol$, the ``virtual users'' approach allows us to define a protocol $\protocol'$ that operates in the weighted model by, essentially, emulating the nominal model, as long as we can solve {\WRFull} with parameters $\fRw := \fAw$ and $\fRn := \fAn$.
If $\fAw < \fAn$, by \Cref{thm:wr-bound}, $T = \sum_{i \in [n]} t_i$ will be at most $O\left(\frac{n}{\fAn - \fAw}\right)$.
In $\protocol'$, each party $i$ participates in $\protocol$ with $t_i$ virtual identities.
Two components of the transformation depend on the problem at hand (but not on the underlying protocol $\protocol$):
\begin{enumerate}
    \item Mapping the input of $i$ in $\protocol'$ to the inputs of its virtual identities in $\protocol$;

    \item Treatment of the outputs of $i$'s virtual identities in $\protocol$ to produce the outputs in $\protocol'$.
\end{enumerate}

We illustrate the black-box transformation with two examples: Validated Byzantine Agreement~\cite{vaba-cachin} and Single Secret Leader Election~\cite{boneh2020single}.

\subparagraph*{\bf Consensus.}

For concreteness, let us consider the problem of Validated Byzantine Agreement (VBA)~\cite{vaba-cachin}.
However, one can easily verify that the same logic will apply to most, if not all, of the many types of consensus and state machine replication, including both crash and Byzantine fault-tolerant ones.
\begin{definition}
    A protocol solves validated Byzantine agreement with external validity predicate $\mathcal{V}$ if it satisfies the following conditions:
    \begin{description}
        \item[Liveness:] Each honest party outputs a value.

        \item[Agreement:] No two honest parties can output different values.

        \item[External Validity:] If an honest party outputs $v$, then $\mathcal{V}(v)$ holds.
        
        \item[Integrity:] If all parties are honest, and if some party decides $v$, then $v$ is the input of some party.
        
        \item[Efficiency:] The communication complexity is probabilistically uniformly bounded.
   \end{description}
\end{definition}

Consider an arbitrary protocol $\protocol$ that solves the problem for some external validity predicate $\mathcal{V}$.
Let $\protocol'$ be the protocol obtained from $\protocol$ by applying the transformation described above with the problem-specific part defined as follows:
\begin{enumerate}
    \item The input of all virtual identities of party $i$ in $\protocol$ is the same as $i$'s input in $\protocol'$;
    \item If $t_i \neq 0$, party $i$ outputs the value output by its first virtual identity and sends it to all parties $j$ such that $t_j = 0$. If $t_i = 0$, it waits for messages from parties with total weight greater than $\fAw W$ vouching for the same output $v$ and outputs $v$.
\end{enumerate}

By construction and the definition of {\WR}, assuming that at most a fraction $\fAw$ of the total weight is corrupted, at most a fraction $\fAn$ of virtual identities will be corrupted and, hence, assuming $\protocol$ solves VBA with nominal resilience $\fAn$, the simulated protocol will satisfy the properties of VBA.
One can easily verify that each of the five properties will be satisfied for $\protocol'$ as well.
Notice, in particular, that efficiency will still be satisfied as the total communication complexity will be increased by only a constant factor (assuming $\fAw$ and $\fAn$ to be constants).

\subparagraph*{\bf Single Secret Leader Election.} SSLE~\cite{boneh2020single} is a distributed protocol that has as an objective to select one of the participants to be a leader with an additional constraint that only the elected party knows the result of the election. Then, once the leader is ready to make a proposal, it reveals itself and other participants can then correctly verify that the claiming leader was indeed elected by the protocol. 

The original paper~\cite{boneh2020single} contains nominal solutions for the protocol relying on ThFHE~\cite{ThFHE-boneh} and on shuffling a list of commitments under the DDH assumption. The authors initially suggest that their protocols could support weights by
replicating each party to match their weights.
This approach is identical to the transformation described in this section with the exception that it does not include weight reduction and, thus, exhibits overhead proportional to the total weight (which can be prohibitively large, see \Cref{tab:evaluations}).
We can solve this issue by applying {\WRFull} at the cost of lowering the resilience by an arbitrarily small constant $\epsilon$ ($\fAw = \fAn - \epsilon$).

However, the original problem definition requires
the election to be \emph{fair}, that is, for the probability of each party being elected to be uniform.
It is easy to see that, as a result of applying weight reduction, this property will not be maintained.
Instead, we can relax it to an alternative property of \emph{chain-quality}, requiring that the fraction of blocks produced by corrupt parties should not surpass a constant fraction $\alpha$ when the adversary might control a fraction of the weights up to $\fAw$. 
Our transformation then trivially solves this problem for $\alpha := \fAn$.

Properties such as \emph{fairness} are one of the limitations of our transformations since any property that is a function of the weight of the parties may not be preserved after the transformation is applied.
We discuss fairness in slightly more detail and speculate about possible fixes to this issue in \Cref{sec:conclusion}.

\section{Applications of Weight Qualification}

\label{sec:wq-applications}

\subsection{Erasure-Coded Storage and Broadcast}

Erasure-coded storage systems~\cite{ida-rabin-89,ida-cachin-tessaro-05,ida-hendricks-07,ida-nazirkhanova-21,ida-yang-22}, also known under the names of Information Dispersal Algorithms (IDA)~\cite{ida-rabin-89} and Asynchronous Verifiable Information Dispersal (AVID)~\cite{ida-cachin-tessaro-05}, are crucial to many systems for space and communication-efficient, secure, and fault-tolerant storage.
Moreover, as demonstrated in~\cite{ida-cachin-tessaro-05}, they can yield highly communication-efficient solutions to the very important problem of asynchronous Byzantine Reliable Broadcast~\cite{textbook,BraTou85}, a fundamental building block in distributed computing that, among other things, serves as the basis for many practical consensus~\cite{honey-badger,beat-bft,all-you-need-is-dag,narwhal,bullshark}, distributed key generation~\cite{adkg,practical-adkg}, and mempool~\cite{narwhal} protocols.

The challenge of applying these protocols in the weighted setting is that $(k,m)$ erasure coding, by definition, converts the original data into $m$ discrete \emph{fragments} such that any $k$ of them are sufficient to reconstruct the original information.
Thus, each party will inevitably get to store an integer number of these fragments, and the smaller $m$ is, the more efficient the encoding and reconstruction will be.
Moreover, for the most commonly used codes--Reed Solomon--the original message must be of size at least $k \log m$ bits. Hence, using a large $m$ may lead to increased communication as the message may have to be padded to reach this minimum size.
As we illustrate in this section, determining the smallest ``safe'' number of fragments to give to each party is exactly the {\WQFull} problem defined in \Cref{sec:problem-statement}.

Let us consider the example of~\cite{ida-cachin-tessaro-05} as it is the first erasure-coded storage protocol tolerating Byzantine faults. We believe {\WQFull} can be applied analogously to other similar works. 

This protocol operates in a model where any $t$ out of $n$ parties can be malicious or faulty, where $t < \frac{n}{3}$.
In other words, it has the nominal fault threshold of $f_n = \frac{1}{3}$.
The protocol encodes the data using $(t+1,n)$ erasure coding, and the data is considered to be reliably stored once at least $2t+1$ parties claim to have stored their respective fragments.
The idea is that, even if $t$ of them are faulty, the remaining $t+1$ parties will be able to cooperate to recover the data.

In order to make a weighted version of this protocol, instead of waiting for confirmations from $2t+1$ parties, one needs to wait for confirmations from a set of parties that together possess more than a fraction $2\fw$ of total weight, where $\fAw = \fAn = \frac{1}{3}$.
A subset of weight less than $\fAw$ of these parties may be faulty.
Hence, for the protocol to work, it is sufficient to guarantee that any subset of total weight more than $2\fAw-\fAw=\fAw$ gets enough fragments to reconstruct the data.
To this end, we can apply the {\WQ} problem with the threshold $\fQw = \fAw$.
We can set $\fQn$ to be an arbitrary number such that $0 < \fQn < \fQw$.
Then, we can use $(\lceil\fQn T\rceil, T)$ erasure coding, where $T$ is the total number of tickets allocated by the {\WQ} solution.
Hence, whenever a set of parties of weight more than $2\fAw$ claim to have stored their fragments, we will be able to reconstruct the data with the help of the correct participants in this set.
As for the rest of the protocol, it can be converted to the weighted model simply by applying weighted voting, as was discussed in \Cref{subsec:weighted-voting}.

As a result, we manage to obtain a weighted protocol for erasure-coded verifiable storage with the same resilience as in the nominal protocol ($\fAw = \fAn = \frac{1}{3}$).
The ``price'' we pay is using erasure coding with a smaller rate ($\fQn$ instead of $\fAw$), i.e., storing data with a slightly increased level of redundancy. However, note that $\fQn$ can be set arbitrarily close to $\fAw$, at the cost of more total tickets and, hence, more computation.

\subsubsection*{Example instantiations} 
The communication and storage complexity of these protocols depends linearly on the rate of the erasure code.
Using Reed-Solomon with Berlekamp-Massey decoding algorithm, the decoding computation complexity~\cite{rs-decoding-complexity} is $O(m^2 \cdot \frac{M}{rm}) = O(\frac{m}{r} \cdot M)$, where $M$ is the size of the message (which we do not affect), $r$ is the rate of the code (in our case, $r = \fQn$), and $m$ is the number of fragments (in our case, the number of tickets allocated by the solution to the {\WQ} problem).
For the sake of illustration, let us fix $\fQn$ to be $\frac{1}{4}$.
Then, the rate of the code used in the weighted solution will be $\frac{4}{3}$ times smaller than in the nominal solution.
For the number of fragments $m$, let us substitute the upper bound from \Cref{col:wq-bound} ($m \le \left\lceil \frac{\fQw(1 - \fQw)}{\fQw - \fQn} n\right\rceil$). 
For $\fQw = \frac{1}{3}$ and $\fQn = \frac{1}{4}$, $m \le \frac{8}{3}n$.
Hence, the overall slow-down compared to the nominal solution is $\frac{8}{3} \cdot \frac{4}{3} \approx 3.56$.

One can also consider using FFT-based decoding algorithms~\cite{rs-decoding-fft}. Since the complexity of the FFT-based decoding depends only polylogarithmically on the number of fragments $m$, one can select the rate of the code ($r = \fQn$) to be much closer to $\fQw$ and, thus, minimize communication and storage overhead.

Some protocols~\cite{rbc-erasures-high-threshold-trusted-setup} are designed for higher reconstruction thresholds, which allows them to be more communication- and storage-efficient compared to~\cite{ida-cachin-tessaro-05}. For these cases, we will need to set $\fQw := \frac{2}{3}$.
By setting $\fQn := \frac{1}{2}$ and applying the upper bound from \Cref{col:wq-bound}, we will obtain the same reduction of factor $\frac{4}{3}$ in rate and $2$ times fewer tickets: $m \le \frac{1/3 \cdot 2/3}{2/3 - 1/2} n = \frac{4}{3}$.
The computational overhead will be $\frac{4}{3} \cdot \frac{4}{3} \approx 1.78$.

\subsection{Error-Corrected Broadcast}

\sloppy
The exciting work of~\cite{dxr21} illustrated how one can avoid the need for complicated cryptographic proofs in the construction of communication-efficient broadcast protocols by employing error-correcting codes, thus enabling a better communication complexity when a trusted setup is not available.
The protocol of~\cite{dxr21} can be used for the construction of communication-efficient Asynchronous Distributed Key Generation~\cite{adkg,practical-adkg} protocols.

Similarly to erasure codes, error-correcting codes convert the data into $m$ discrete fragments, such that any $k$ of them are sufficient to reconstruct the original information.
However, they have the additional property that the data can be reconstructed even when some of the fragments input to the decoding procedure are invalid or corrupted.
Reed-Solomon decoding allows correcting up to $e$ errors when given $k+2e$ fragments as input.

The protocol of~\cite{dxr21} tolerates up to $t$ failures in a system of $n \ge 3t+1$ parties (for simplicity, we will consider the case $n = 3t+1$).
Its key contribution is the idea of \emph{online error correction}. Put simply, the protocol first ensures that:
\begin{itemize}
    \item Every honest party obtains a cryptographic hash of the data to be reconstructed;
    \item Every honest party obtains its chunk of the data.
\end{itemize}
Then, in order to reconstruct a message, an honest party solicits fragments from all other parties and repeatedly tries to reconstruct the original data using the Reed-Solomon decoding and verifies the hash of the output of the decoder against the expected value.
As the protocol uses $k = t+1$ and $m = n$, after hearing from all $2t+1$ honest and $e \le t$ malicious parties, it will be possible to reconstruct the original data (as $2t+1+e \ge k + 2e$, for $k=t+1$).

To convert this protocol into the weighted model, it is sufficient to make sure that all honest parties together possess enough fragments to correct all errors introduced by the corrupted parties.
To this end, we will apply the {\WQ} problem.
We will set $\fQw$ to the fraction of weight owned by honest parties, i.e., $\fQw := 1 - \fAw = \frac{2}{3}$ (where $\fAw$ will be the resilience of the resulting weighted protocol, $\fAw = \fAn = \frac{1}{3}$).
However, it is not immediately obvious how to set $\fQn$ to allow the above-mentioned property.

If we want to use error-correcting codes with rate $r$, we need to guarantee that the fraction of fragments received by the honest parties (which is at least $\fQn$) is at least $r+e$, where $e$ is the fraction of fragments received by the corrupted parties.
However, since honest parties get at least the fraction $\fQn$ of all fragments, then $e \le 1 - \fQn$.
Hence, we need to set $\fQn$ so that $\fQn \ge r + (1 - \fQn)$.
We can simply set $\fQn := \frac{r}{2} + \frac{1}{2}$ for arbitrary $r < \frac{1}{3}$.

\subsubsection*{Example instantiation}
For the sake of an example, we can set $\fQw := \frac{2}{3}$, $r := \frac{1}{4}$ and $\fQn := \frac{5}{8}$.
Then, using the bound from \Cref{col:wq-bound}, the number of tickets will be at most $\frac{2/3 \cdot 1/3}{2/3 - 5/8} \cdot n \le \frac{16}{3}n$.

As was discussed above, for erasure codes, we can either use the Berlekamp-Massey decoding algorithm or the FFT-based approaches.
The same applies to error-correcting codes.
As most practical implementations use the former, we will focus on it.
In this case, the communication overhead will be $\frac{r_n}{r_w}$, where $r_n = \frac{1}{3}$ is the rate used in the nominal protocol and $r_w$ is the rate used for the weighted protocol (in the example above, $r = \frac{1}{4}$).
The computation overhead is $\frac{r_n}{r_w} \cdot \frac{T}{n}$, where $T$ is the number of tickets allocated by the {\WQ} solution (in the example above, $T \le \frac{16}{3}n$ in the worst case).
Hence, for the example parameters, the worst-case computational overhead is $\frac{4}{3} \cdot \frac{16}{3} \approx 7.11$.

\section{Derived Applications} \label{sec:derived-applications}

In this section, we discuss indirect applications of weight reduction problems that are obtained by using one or multiple building blocks discussed in \Cref{sec:wr-applications,sec:wq-applications}.
For all applications discussed here, we manage to avoid losing resilience despite applying weight reduction.
In all cases, the majority of the protocol logic should be converted to the weighted model by applying weighted voting, as discussed in \Cref{subsec:weighted-voting}. 

\subsection{Asynchronous State Machine Replication}

For asynchronous state machine replication protocols~\cite{honey-badger,beat-bft,all-you-need-is-dag,narwhal,bullshark}, we simply need to use a weighted communication-efficient broadcast protocol (discussed in \Cref{sec:wq-applications}) and weighted distributed random number generation (discussed in \Cref{sec:random-beacon}).
distributed number generation part can use a nominal protocol with threshold $\fRn = \frac{1}{2}$ and set $\fRw := \frac{1}{3}$, which is the resilience of the rest of the protocol.
Thus, in some sense, we level the resilience of different parts of the protocol, without affecting the resilience of the composition.

\subsection{Validated Asynchronous Byzantine Agreement}

The same approach can be applied to generate randomness for Validated Asynchronous Byzantine Agreement (VABA)~\cite{vaba-cachin,vaba-abraham}.

These protocols also require tight threshold signatures.
However, in practice, multi-signatures~\cite{ohta1999multi,ThresholdBLS} are usually applied instead as they have almost no overhead over threshold signatures on the system sizes where such protocols could be applied (below 1000 participants): it suffices to append the multi-signature with an array of $n$ bits, indicating the set of parties that produced the signature.
Then, along with the verification of the validity of the multi-signature itself, anyone can verify that the signers together hold sufficient weight.

Alternatively, one could apply the approach described in \Cref{subsec:tight-secret-sharing} to implement tight weighted threshold signatures. However, it would lead to an increase in message complexity of the resulting protocol, which we want to avoid.

Finally, an ad-hoc weighted threshold signature scheme can be applied, such as the one recently proposed in~\cite{das2023threshold}.
Note that these signatures cannot be used for distributed randomness generation as they lack the necessary uniqueness property, and thus we still need to apply {\Swiper} to obtain a complete protocol.

\subsection{Consensus with Checkpoints}

We can apply the same approach for checkpointing proof-of-stake consensus protocols~\cite{pikachu}, but this time for blunt threshold signatures (as discussed in \Cref{subsec:blunt}) instead of random number generation. If, for some reason, one wants to use a tight threshold signature, the approach described in \Cref{subsec:tight-secret-sharing} can be applied at the cost of just 1 additional message delay per checkpoint.

Compared to ad-hoc solutions for weighted threshold signatures~\cite{das2023threshold}, we claim that our approach is more computationally efficient as it is basically as fast as the underlying nominal protocol. 
For example, 2 pairings to verify a BLS signature~\cite{ThresholdBLS} compared to 13 pairings to verify a signature in~\cite{das2023threshold}.
Moreover, the weight reduction approach is more general and can support other types of threshold signatures, such as RSA~\cite{rsa-threshold-signatures} and Schnorr~\cite{ThresholdSchnorr}, the latter being particularly important in the context of checkpointing to Bitcoin~\cite{pikachu}.

\section{Analyzing {\WRFull} on sample systems} \label{sec:empirical-study}

\myparagraph{Data sets}
We analyzed our protocol using four real-world data sets for weight distribution: Aptos~\cite{aptos_white, aptos_stake}, Tezos~\cite{tezos_white, tezos_stake}, Filecoin~\cite{filecoin_white, filecoin_stake}, and Algorand~\cite{algorand_white, algorand_stake}.
For the reader's convenience, we provide the results for all the datasets in a separate \cref{sec:remainingplots} and present the results for only one blockchain (Tezos) in \Cref{fig:exptezos} as an example.

\begin{figure*}
    \centering
    \includegraphics[scale=0.6]{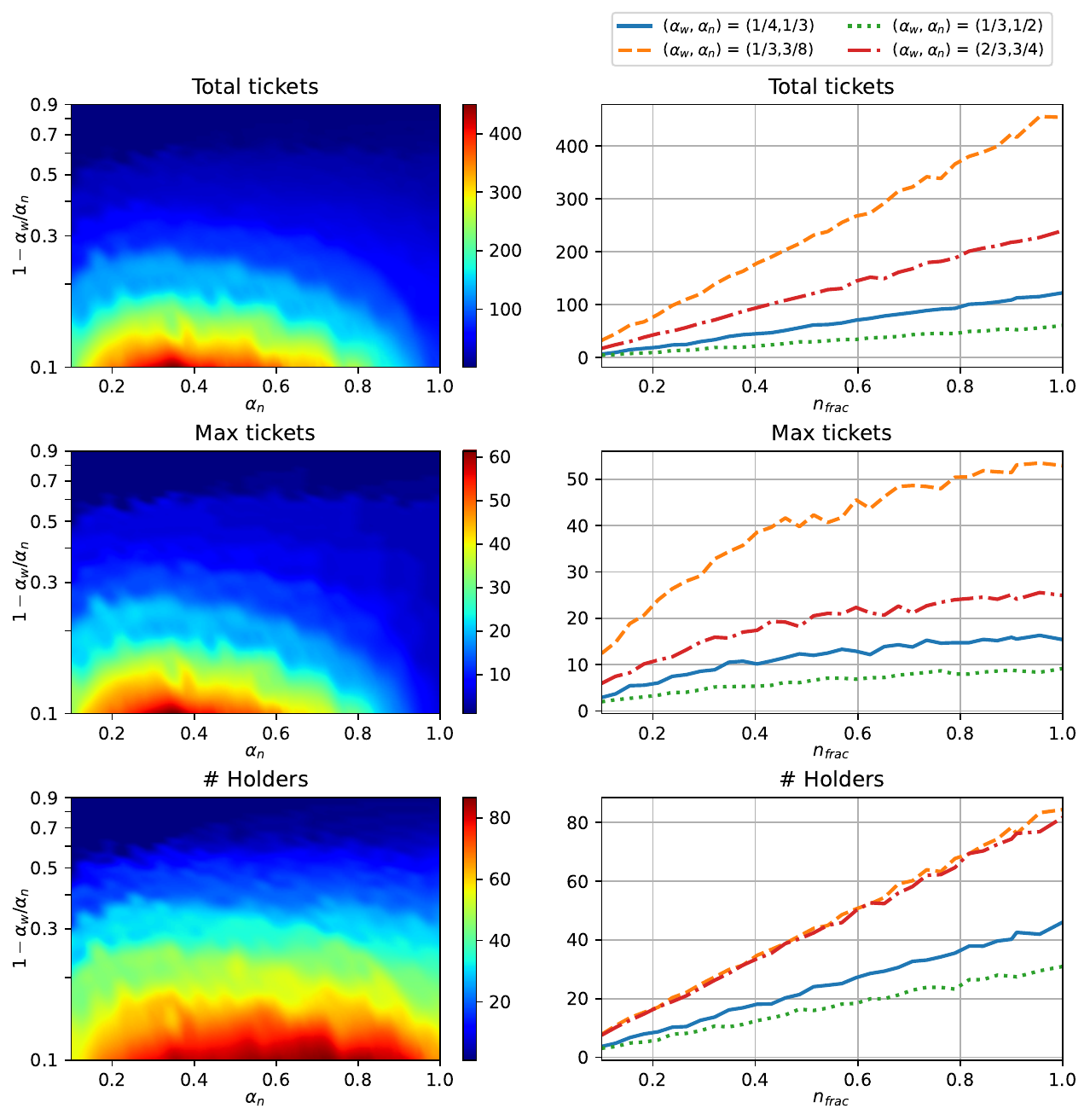}
        \caption{Experiment results using Tezos}
    \label{fig:exptezos}
\end{figure*}

\myparagraph{Experiment description} 
We performed two kinds of experiments on real blockchain data. In the first experiment, shown in the left column of \Cref{fig:exptezos}, we analyzed the influence of the choice of parameters $\fRw$ and $\fRn$ for the original data retrieved from the blockchains; the value of $\fRn$ was varied in the range $[0.1, 1]$, while the value of $\fRw$ was tested in the range $[0.1 \times \fRn, 0.9 \times \fRn]$. In the experiments showcased in the right column of \Cref{fig:exptezos}, we kept these parameters fixed and analyzed the influence of the number of parties in the metrics we tracked. In order to simulate having the same blockchain with different numbers of parties, we used the statistical technique known as bootstrapping. To this end, we performed $100$ experiments sampling parties with replacement from the blockchain data and taking the average of the results.

In each experiment, we tracked the total number of tickets distributed, the maximum number of tickets held by a single party, and the number of parties that get at least one ticket (in the figures, we label them as the number of holders). In \Cref{fig:exptezos}, we show the results for the Tezos blockchain. The results for Algorand, Aptos, and Filecoin are available in \Cref{sec:remainingplots}. The analysis of the results reveals the following information: the upper bound given in \Cref{sec:problem-statement} is very pessimistic for weight distributions emerging in practice, with the total number of tickets rarely surpassing the number of parties for different values of $\fRn$ and $\fRw$. The total number of tickets varies extremely close to a linear function on the number of parties, as well as the number of holders. The maximum number of tickets, on the other hand, seems to saturate when the number of parties in absolute terms surpasses the order of magnitude of $1000$, remaining almost constant after that point.

\section{Related Work} \label{sec:related}

\myparagraph{Knapsack}
The Knapsack problem and its variations hold huge importance in theoretical computer science and have numerous applications in both theory and practice.
The weight reduction problems studied in this paper seem to be related to, or can even be seen as a variation of the famous Knapsack problem.
For example, one can see {\WRFull} as the problem of constructing ``worst possible'' profits for a Knapsack instance given the weights and the capacity.
We refer to~\cite{KnapsackProblems04} for a comprehensive survey on the topic.

\myparagraph{Virtual users} 
The simplest solution for creating a weighted threshold cryptographic system is to simply have a user of weight $w$ become $w$ virtual users and to give one key to each of them.
Shamir's paper describing his secret sharing scheme~\cite{Sha79} puts forward this solution.
However, in practice, the total weight tends to be prohibitively large, and ``quantizing'' it requires solving weight reduction problems, which is the main subject of this paper.

\myparagraph{Weighted voting}
In~\cite{weighted-voting}, Gifford presents the idea of weighted voting for distributed storage systems.
The paper suggests assigning weights to replicas according to the estimated failure probabilities and using weight-based quorums to store and retrieve data.
We discuss the merits and limitations of this approach in \Cref{subsec:weighted-voting}.
The goal of our paper is to complement the weighted voting approach and design a framework for implementing weighted distributed protocols that can benefit from solutions and primitives that are initially designed for the nominal model.
In \Cref{sec:wr-applications,sec:wq-applications,sec:derived-applications}, we discuss in detail how to combine weighted voting and weight reduction to obtain extremely efficient weighted protocols without sacrificing resilience.

\myparagraph{Ad-hoc solutions}
There is a large body of work studying ad-hoc weighted cryptographic protocols~\cite{GarJaiMukSinWanZha22,BeiWei06,ChaKia21,ItoSaiNis89,das2023threshold,ss-from-wiretap-channels}.
Compared to these works, the weight reduction approach studied in this paper has a number of benefits, such as simplicity, efficiency, wider applicability, and a wider range of possible cryptographic assumptions.
Moreover, in many cases, ad-hoc solutions can be combined with and benefit from weight reduction.
In this paper, we also study other, non-cryptographic, applications, such as erasure and error-corrected distributed storage and broadcast protocols.

\myparagraph{Similar work by Benhamouda, Halevi, and Stambler}
A recent work~\cite{ss-from-wiretap-channels} mentioned a similar idea of reducing real weights to integers to construct \emph{ramp} secret sharing.
This project has been started and the first version of {\Swiper} has been drafted before the online publication and without any knowledge of~\cite{ss-from-wiretap-channels}.
As the main focus of~\cite{ss-from-wiretap-channels} is different, we believe that we do a much more in-depth exploration of this direction by studying different kinds of weight reduction problems and their applications beyond secret sharing, as well as providing much tighter bounds and implementing a solver that is not only linear in the worst case but also allocates very few tickets in empirical evaluations on real-world weight distributions.

\myparagraph{Application in Aptos blockchain}
A version of the {\WSFull} problem has been recently used in the Aptos blockchain in their implementation of on-chain randomness~\cite{das2024distributed}.
They consider an inverse problem where the number of tickets is fixed and the gap between $\alpha$ and $\beta$ is minimized.
Note that one can trivially reduce one problem to the other (in both directions) by using a binary search.

\section{Concluding remarks and future work directions} \label{sec:conclusion}

In this paper, we have presented a family of optimization problems called weight reduction that, to the best of our knowledge, has not been studied before. We provided practical protocols to find good, albeit not optimal, solutions to these problems.
As we have shown, it allows us to obtain efficient implementations of many weighted distributed protocols.

We believe that weight reduction problems will play an important role in the future of blockchain systems as they become more sophisticated and the need for threshold cryptography as well as erasure coding and protocols like single secret leader election grows.
At the time of writing, at least one major layer-1 blockchain has already integrated a version of {\WSFull} for generating on-chain randomness.

In this paper, we attempted the first systematic study of this family of problems, but there are still many important questions being left for future research.

\myparagraph{Fairness}
Weight reduction naturally leads to slight deviations in the relative weights of the participants.
While in this paper we focused on \emph{safety} and \emph{liveness} properties and showed that they can still be preserved, we did not consider any kind of \emph{fairness} properties.
However, we believe that, somewhat counterintuitively, some form of fairness can be preserved as well.
To this end, we are considering two possible directions:
\begin{enumerate}
    \item \textbf{Expected fairness:} In addition to deterministically assigned tickets, we can allocate some small number of tickets randomly so that each party gets exactly the same fraction of tickets as its fraction of weight \emph{in expectation}. We believe that it can be done while still preserving safety and liveness \emph{deterministically}, i.e., even in the worst case, when all the ``random'' tickets are received by the adversary.

    \item \textbf{Integral fairness:} Similarly, one can imagine a \emph{deterministic} protocol that provides fairness \emph{over time}.
    In such a scheme, the ticket assignment will be updated periodically and each party will get exactly the right number of tickets \emph{on average}, over a large enough period.
\end{enumerate}

\myparagraph{Incentives}
One important aspect of proof-of-stake blockchains is the distribution of incentives, which should depend on the weight of each party. It is not immediately clear what is the right way to allocate incentives in a system where weight reduction is being applied.

\myparagraph{Other applications}
While we covered a wide range of applications in this paper, we believe that there must be others, including ones not related to distributed computing or applied cryptography.

\myparagraph{Adversarial attacks}
In this paper, we study the ``worst case'' weight distributions by providing the upper bounds and the ``organic case'' by studying the real-world weight distributions. However, in practice, under an adversarial attack, the weight distribution will be a hybrid one: the weights of honest parties will be organic, but the weights of the adversarial parties may be redistributed maliciously. It is an interesting avenue for future work to study how much an adversary can affect the number of tickets (and, thus, the performance of the system) by redistributing their weight in a malicious manner.

\myparagraph{Complexity and more precise bounds}
Finally, there are still many theoretical questions about these problems.
Do they have polynomial-time exact solutions?
What are the lower bounds?
Can we derive better upper bounds?
Moreover, what are some other interesting and useful weight reduction problems, apart from the three defined in this paper?

\balance
\section*{Acknowledgement}
    We are grateful to Benny Pinkas for recommending the inclusion of a constant in the construction of {\Swiper}, a suggestion that helped us significantly reduce the number of allocated tickets both in theory and in practice,
    and to the anonymous PODC reviewers for constructive feedback on the paper structure and presentation.
    Andrei Tonkikh is supported by TrustShare Innovation Chair (financed by Mazars). Luciano Freitas is supported by Nomadic Labs.

\clearpage
\bibliographystyle{plainurl}
\bibliography{references}

\clearpage
\appendix
\nobalance

\section{Proofs} \label{sec:proof}

In this section, we provide formal proofs for \Cref{thm:wr-bound,col:wq-bound,thm:ws-bound}.

\subsection{Upper bounds on {\WRFull} and {\WSFull}} \label{subsec:wr-bound}

Let us start with some auxiliary definitions.
A \emph{ticket assignment} $t$ is a vector of $n$ numbers: $t_1, \dots, t_n \in \mathbb{Z}_{\ge 0}$.
With a slight abuse of notation, for a ticket assignment $t$ and a set $S \subseteq [n]$, we use notation $t(S)$ to denote $\sum_{i \in S} t_i$.
Let us say that a ticket assignment $t$ is \emph{viable} if $t([n]) \neq 0$ and $\forall S \subseteq [n]:$ if $w(S) < \fRw W$, then $t(S) < \fRn t([n])$, that is if it satisfies the requirements of the {\WRFull} problem as defined in \Cref{sec:problem-statement}.

In this section, we formally prove \Cref{thm:wr-bound} by constructing a viable ticket assignment $\hat{t}$ such that $\hat{t}([n]) \le
\left\lceil \frac{\fRw (1 - \fRw)}{\fRn - \fRw} n\right\rceil$.
As the starting point, we consider a family of ticket assignments parameterized by a single number $s > 0$:
\begin{center}
    $(t_{s})_i := \left\lfloor w_i s + \fRw \right\rfloor$.
\end{center}

Let $s^*$ be a \emph{locally minimal} viable value for $s$, i.e., a positive number such that $t_{s^*}$ is viable, but $t_{s^*-\varepsilon}$ is not, for any sufficiently small $\varepsilon$.
Since we already proved that viable values of $s$ exist, it is easy to see that such $s^*$ exists.
Moreover, there must be some $j$ such that $s^* w_j + \fRw$ is an integer.
Indeed, if this does not hold, we would be able to slightly decrease $s^*$ without changing the ticket assignment, which would contradict the assumption that $s^*$ is a local minimum.
Let $t^* := t_{s^*}$ and $J := \{j \in [n] \mid s^* w_i + \fRw\text{ is an integer}\}$.
Let $t'$ be a ticket assignment in which we take one ticket from each party in $J$, i.e.: 

\begin{align*}
t'_{i} := 
    \begin{cases}
        t^*_{i} - 1 & \text{ if $i \in J$} \\
        t^*_{i} & \text{ otherwise}
    \end{cases}
\end{align*}

Notice that $t'$ is equal to $t_{s^* - \varepsilon}$ for a sufficiently small $\varepsilon > 0$.\footnote{Indeed, if we decrease $s^*$ by any positive amount, each party in $J$ will lose at least one ticket as they will step over the rounding threshold.
However, it is also easy to see that $\varepsilon$ can be made small enough so that no other party will lose a ticket and no party in $J$ will lose more than one ticket.}
Hence, by construction, $t'$ is not viable.
Now, let us consider a set of ``intermediate'' ticket assignments:
we will be taking tickets from parties in $J$ as long as the ticket assignment stays viable.
We will end up with two ticket assignments: $\hat{t}$ and $\doublehat{t}$ such that $\hat{t}$ is viable and $\doublehat{t}$ is not, and $\hat{t}([n]) = \doublehat{t}([n]) + 1$.
All that is left is to prove that $\hat{t}([n]) \le \left\lceil \frac{\fRw (1 - \fRw)}{\fRn - \fRw} n \right\rceil$ or, equivalently, that $\doublehat{t}([n]) \le \left\lceil \frac{\fRw (1 - \fRw)}{\fRn - \fRw} n \right\rceil - 1$.

Since $\doublehat{t}$ is not viable, either $\doublehat{t}([n]) = 0$ or there must exist a set $S \subseteq [n]$ such that $w(S) < \fRw W$ and $\doublehat{t}(S) \ge \fRn \doublehat{t}([n])$.
As the former case is trivial, we will focus on the latter.
Let us provide an upper bound on $\doublehat{t}(S)$ and a lower bound on $\doublehat{t}(\overline{S})$, where $\overline{S} := [n] \setminus S$.
To this end, let us note that, for any $i \in [n]$, it holds that $\doublehat{t}_i \ge w_i s^* + \fRw - 1$.
Indeed, there are two cases to consider:
\begin{enumerate}
    \item if $\doublehat{t}_i = t^*_i$, the inequality holds trivially as $\doublehat{t}_i = t^*_i = \lfloor w_i s^* + \fRw \rfloor$;
    \item otherwise, $\doublehat{t}_i = t^*_i - 1$. However, by construction, it means that $w_i s^* + \fRw$ is an integer and, thus $t^*_i = w_i s^* + \fRw$ and $\doublehat{t}_i = w_i s^* + \fRw - 1$.
\end{enumerate}
Hence:
\begin{align*}
    \doublehat{t}(S) 
        &=   \sum_{i \in S} \doublehat{t}_i \\
        &\le \sum_{i \in S} t^*_i =   \sum_{i \in S} \lfloor w_i s^* + \fRw \rfloor \\
        &<   \fRw W s^* + \fRw |S| \\
    \doublehat{t}(\overline{S})
        &=   \sum_{i \not\in S} \doublehat{t}_i \\
        &\ge \sum_{i \not\in S} (w_i s^* + \fRw - 1) \\
        &>   (1 - \fRw) W s^* - (1 - \fRw) (n - |S|)
\end{align*}

By construction, $\doublehat{t}(S) \ge \fRn \doublehat{t}([n])$ and
$\doublehat{t}([n]) = \doublehat{t}(S) + \doublehat{t}(\overline{S})$.
Hence, $(1 - \fRn) \doublehat{t}(S) \ge \fRn \doublehat{t}(\overline{S})$.
From this, we can derive an upper bound on $s^*$:
\begin{align*}
    & (1 - \fRn) \doublehat{t}(S) \ge \fRn \doublehat{t}(\overline{S}) \Rightarrow \\
    \Rightarrow\; & (1 - \fRn)(\fRw W s^* + \fRw |S|) \\
    & \qquad > \fRn((1 - \fRw) W s^* - (1 - \fRw) (n - |S|)) \\
    \Rightarrow\; & s^* < \frac{\fRn (1 - \fRw) n}{(\fRn - \fRw)W} - \frac{|S|}{W}
\end{align*}

Finally, we can combine everything into an upper bound on $\doublehat{t}([n])$:
\begin{align*}
    \doublehat{t}([n]) 
        &\le \frac{\doublehat{t}(S)}{\fRn} \\
        &<   \frac{\fRw}{\fRn} (W s^* + |S|) \\
        &<   \frac{\fRw}{\fRn} \left(\frac{\fRn (1 - \fRw) n}{\fRn - \fRw} - |S| + |S|\right) \\
        &=   \frac{\fRw (1 - \fRw)}{\fRn - \fRw} n
\end{align*}

Since $\doublehat{t}([n])$ is an integer and the inequality is strict, we can rewrite it as $\doublehat{t}([n]) \le \left\lceil \frac{\fRw (1 - \fRw)}{\fRn - \fRw} n \right\rceil - 1$.
As, by construction, $\hat{t}$ is viable and $\hat{t}([n]) = \doublehat{t}([n]) + 1$,
we found a viable ticket assignment with at most $\left\lceil \frac{\fRw (1 - \fRw)}{\fRn - \fRw} n\right\rceil$ tickets, thus concluding the proof of \Cref{thm:wr-bound,col:wq-bound}. \QED

\subsection{Upper bound on {\WSFull}}

Let $\gamma := \frac{\alpha + \beta}{2}$.
For {\WSFull}, we analyze a family of ticket assignments of form 
$t_{s,i} := \lfloor w_i s + \gamma \rfloor$.
Let us consider the case when the {\WSShort} conditions are violated, i.e., there exist sets $S_1$ and $S_2$ such that $w(S_1) < \alpha W$, $w(S_2) > \beta W$, and $t(S_1) \ge t(S_2)$.
This means that at least one of two events happened: $t(S_1) \ge \gamma T$ or $t(S_2) < \gamma T$, or, equivalently, $t(\overline{S_2}) > (1 - \gamma) T$. 
Let us first consider the case when $t(S_1) \ge \gamma T$.
This can only happen when $s < \frac{\gamma(1 - \gamma) n}{(\gamma - \alpha) W}$. The proof is done using the same set of techniques as in \Cref{subsec:wr-bound}:
\begin{align*}
    t(S_1) 
        &=   \sum_{i \in S_1} t_i 
        =   \sum_{i \in S_1} \lfloor w_i s + \gamma \rfloor \\
        &\le \sum_{i \in S_1} (w_i s + \gamma) \\
        &<   \alpha W s + \gamma |S_1| \\
    t(\overline{S_1})
        &=   \sum_{i \not\in S_1} \lfloor w_i s + \gamma \rfloor \\
        &\ge \sum_{i \not\in S_1} (w_i s + \gamma - 1) \\
        & >   \beta W s - (1 - \gamma)(n - |S_1|)
\end{align*}
\begin{align*}
    &t(S_1) \ge \gamma T \\
    \Leftrightarrow\;& (1 - \gamma)t(S_1) \ge \gamma t(\overline{S_1}) \\
    \Rightarrow\;& (1 - \gamma) (\alpha W s + \gamma |S_1|)
            > \gamma (\beta W s - (1 - \gamma)(n - |S_1|)) \\
    \Rightarrow\;& s < \frac{\gamma (1 - \gamma) n}{(\gamma - \alpha)W}
\end{align*}

Analogously, in the case when $t(\overline{S_2}) > (1 - \gamma) T$, we can prove\atadd{ (by substitution of $(1-\gamma)$ in place of $\gamma$ and $(1-\beta)$ in place of $\alpha$)} that:
\begin{gather*}
    s < \frac{(1 - \gamma) (1 - (1 - \gamma)) n}{((1 - \gamma) - (1 - \beta))W} = \frac{\gamma (1 - \gamma) n}{(\beta - \gamma)W}
\end{gather*}
We specifically chose $\gamma = \frac{\alpha + \beta}{2}$ so that the two bounds coincide: $s < \frac{2 \gamma (1 - \gamma) n }{(\beta - \alpha) W}$.
Hence, it is sufficient to select $s := \frac{\gamma (2 - \alpha - \beta) n }{(\beta - \alpha) W}$ to guarantee that neither of the two events happens and $t(S_1) < \gamma T \le t(S_2)$.

Let us now compute a bound on the total number of tickets:
\begin{align*}
    T \le s W + \gamma n = \frac{(\alpha + \beta) (1 - \alpha)}{\beta - \alpha} n
\end{align*}

\QED

\section{Exact solution using MIP} \label{sec:mip}

The way we formulate $\WR$ in \cref{def:wrp} can be directly translated into an instance of bi-level optimization problem~\cite{ColMarSav07}. In such problems, we define an \emph{upper level} optimization problem which contains another (lower-level) optimization problem in its constraints, namely: 

\begin{align*}
\text{minimize} & \sum_{i=1}^nt_i\\
\text{subject to} & \sum_{i=1}^nx_it_i < \fRn \sum_{i=1}^nt_i\\
& \text{maximize} \sum_{i=1}^nx_it_i \\
& \text{subject to} \sum_{i=1}^nw_ix_i < \fRw\sum_{i=1}^nw_i \\
& \sum_{i=1}^nt_i \ge 1\\
& x_i \in \{0,1\}, t_i \in \{0,1,2,\dots\}
\end{align*}

Noticing that the inner optimization problem is the Knapsack problem, we can hard-code a dynamic programming by profits solution to the Knapsack problem into the constraints.
Unfortunately, the resulting MIP has a lot, albeit a polynomial number, of constraints and, thus, is prohibitively slow for inputs of size larger than a couple of dozens.

\section{Experiment Results}
\label{sec:remainingplots}

\Crefrange{fig:exp-tezos-full}{fig:exp-algorand-full} demonstrate the results of the experiments on the data from the stake distribution of 4 major blockchain systems: Aptos~\cite{aptos_white, aptos_stake}, Tezos~\cite{tezos_white, tezos_stake}, Filecoin~\cite{filecoin_white, filecoin_stake}, and Algorand~\cite{algorand_white, algorand_stake}.
The analysis of the experimental results is presented in \Cref{sec:empirical-study}.

\begin{figure*}
    \centering
    \includegraphics[width=\linewidth]{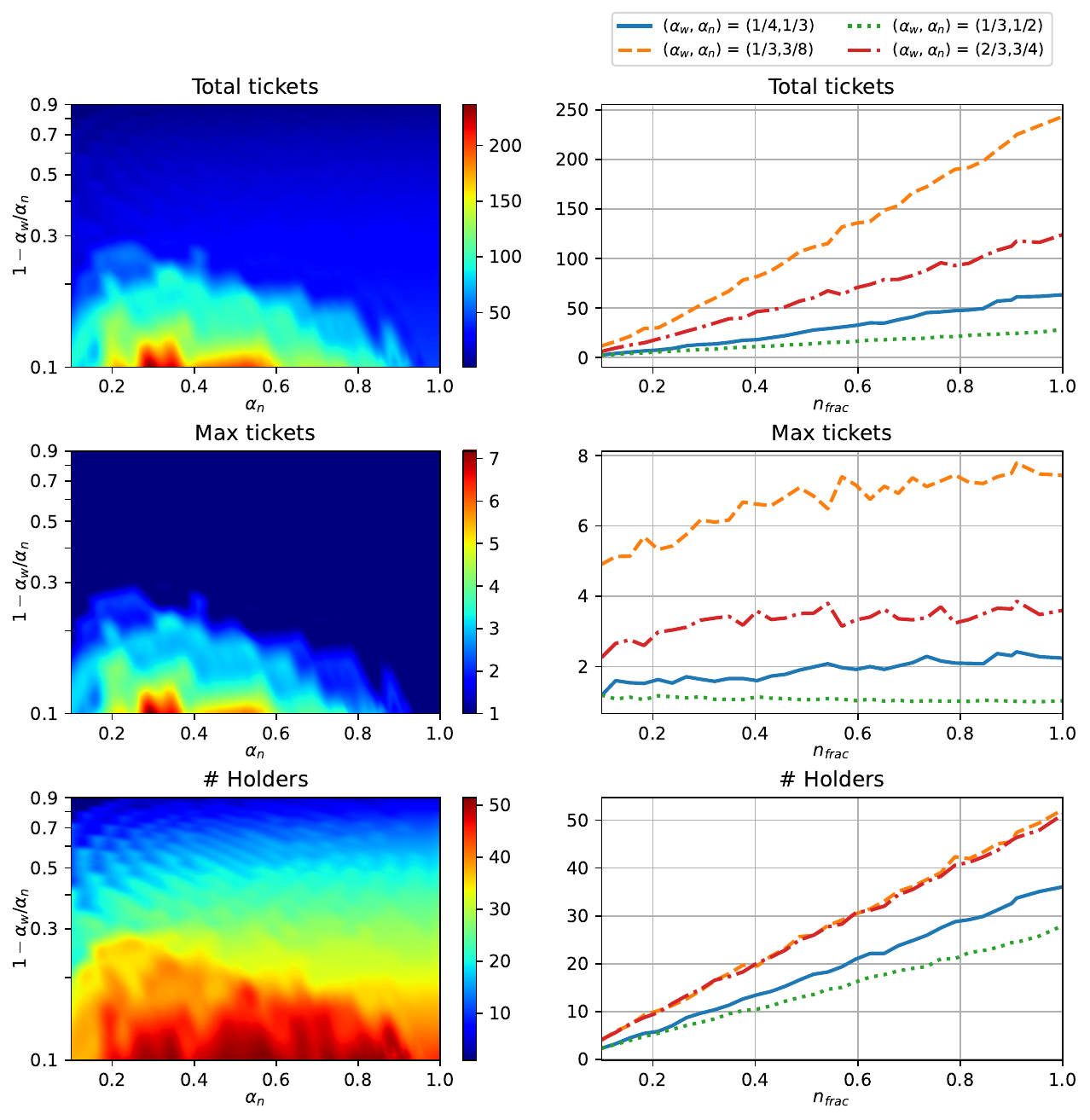}
        \caption{Experiment results using Aptos}
    \label{fig:exp-aptos-full}
\end{figure*}

\begin{figure*}
    \centering
    \includegraphics[width=\linewidth]{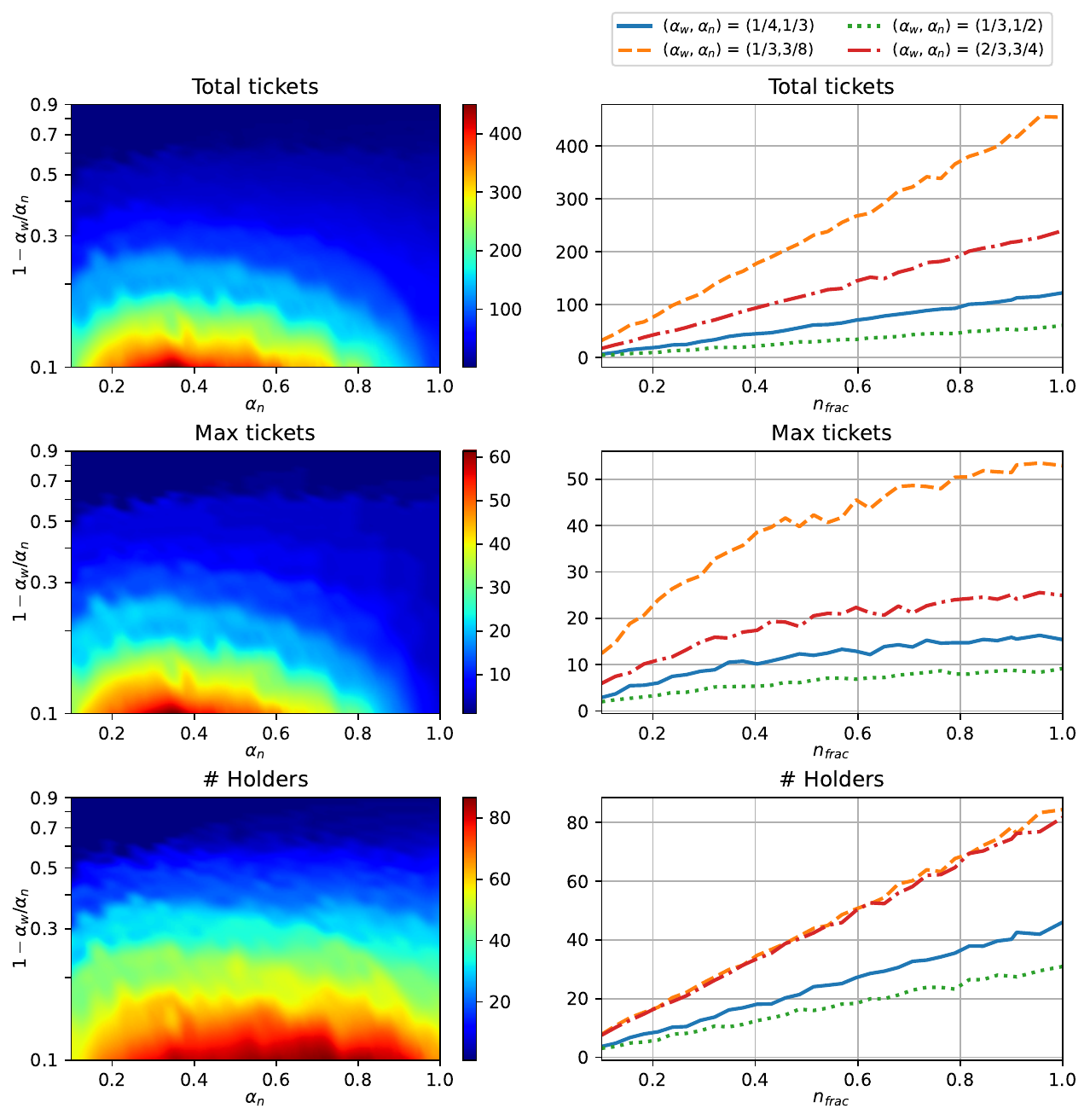}
        \caption{Experiment results using Tezos}
    \label{fig:exp-tezos-full}
\end{figure*}

\begin{figure*}
    \centering
    \includegraphics[width=\linewidth]{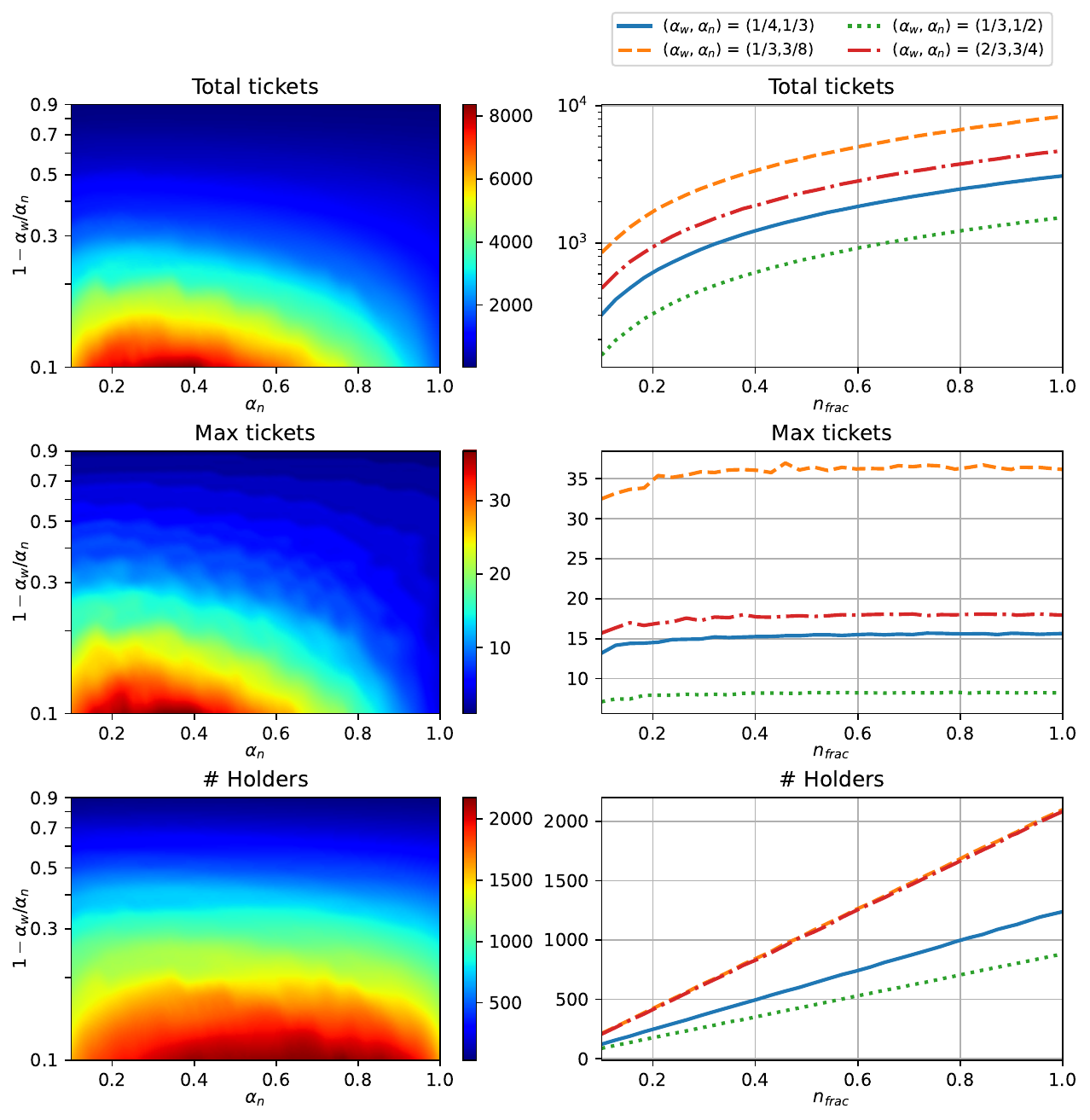}
        \caption{Experiment results using Filecoin}
    \label{fig:exp-filecoin-full}
\end{figure*}

\begin{figure*}
    \centering
    \includegraphics[width=\linewidth]{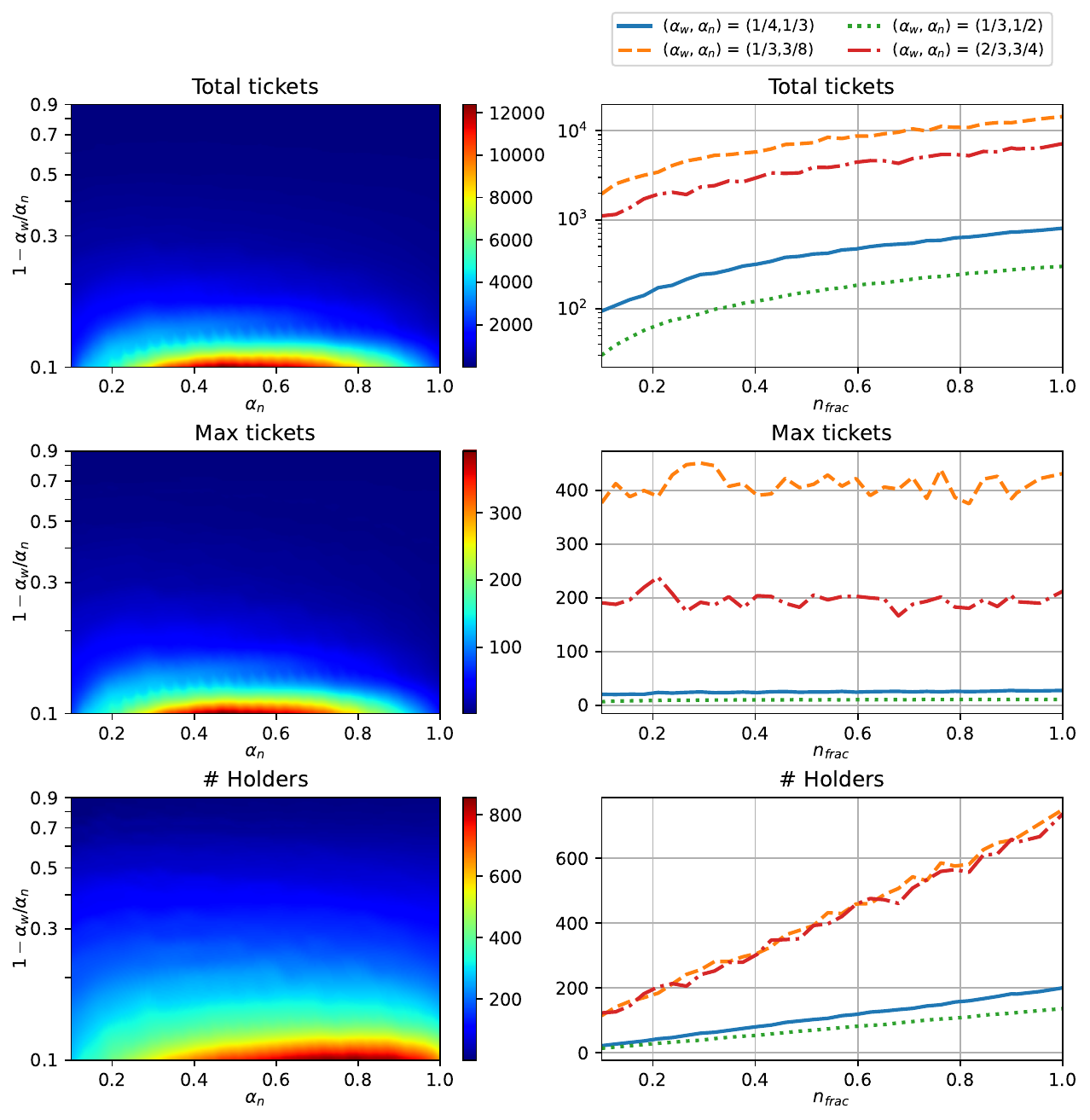}
        \caption{Experiment results using Algorand}
    \label{fig:exp-algorand-full}
\end{figure*}

\end{document}

%% file: definitions.tex
\NewDeclarationType{Function}{\textsc}
\NewDeclarationType{Variable}{\mathit}
\NewDeclarationType{Set}{\mathcal}
\NewDeclarationType{Constant}{\mathbf}
\NewDeclarationType{Problem}[1]{%
    \ifmmode%
        \textnormal{\textsc{#1}}%
    \else%
        #1%
    \fi%
}

\algblockdefx{Parameters}{EndParameters}{\textbf{parameters:}}{}
\algtext*{EndParameters}

\DeclareGlobalSet*[AllProc]{\ensuremath{\Pi}}
\DeclareGlobalSet[Split]{S}
\DeclareGlobalConstant{True}
\DeclareGlobalConstant{False}

\newcommand{\Swiper}{Swiper}

\newcommand{\WRFull}{Weight Restriction}
\newcommand{\WQFull}{Weight Qualification}
\newcommand{\WSFull}{Weight Separation}

\DeclareGlobalProblem{WR}
\DeclareGlobalProblem{WQ}
\DeclareGlobalProblem{WS}

\newcommand{\WRShort}{\WR}
\newcommand{\WQShort}{\WQ}
\newcommand{\WSShort}{\WS}

\newcommand{\fQw}{\beta_w}
\newcommand{\fQn}{\beta_n}
\newcommand{\fRw}{\alpha_w}
\newcommand{\fRn}{\alpha_n}
\newcommand{\fAw}{f_w}
\newcommand{\fAn}{f_n}

\newcommand{\fw}{\fAw}

\newcommand{\protocol}{\mathcal{P}}
\newcommand{\problem}{\Pi}

\newcommand{\myparagraph}[1]{\subparagraph*{\textbf{#1.}}}

\definecolor{sky}{RGB}{50,150,255}
\definecolor{crimson}{RGB}{200,20,60}

\DeclareGlobalVariable[Kamin]{K_{A_\textit{min}}}
\DeclareGlobalVariable[esopt]{e_{s_{opt}}}
\DeclareGlobalVariable[sopt]{s_{opt}}
\DeclareGlobalVariable[Kprec]{K_\mathit{prec}}
\DeclareGlobalFunction{Push}
\DeclareGlobalFunction{Pop}

\algblockdefx[Algorithm]{Algorithm}{EndAlgorithm}[2]{#1(#2)}{}
\algtext*{EndAlgorithm}

\algblockdefx[Input]{Input}{EndInput}{\textbf{Input:}}{}
\algtext*{EndInput}
\algnewcommand{\InputInline}[1]{\textbf{Input:} #1}

\algblockdefx[Output]{Output}{EndOutput}{\textbf{Output:}}{}
\algtext*{EndOutput}
\algnewcommand{\OutputInline}[1]{\textbf{Output:} #1}

\algblockdefx[Operation]{Operation}{EndOperation}[1]{\textbf{Operation} #1:}{}
\algtext*{EndOperation}

\algblockdefx[Upon]{Upon}{EndUpon}[1]{\textbf{Upon} #1:}{}
\algtext*{EndUpon}

\algblockdefx[Variable]{Variable}{EndVariable}{\textbf{Variables:}}{}
\algtext*{EndVariable}

\algblockdefx[Objective]{Objective}{EndObjective}{\textbf{Objective:}}{}
\algtext*{EndObjective}

\algblockdefx[ForConstraint]{ForConstraint}{EndForConstraint}[1]{$\forall #1:$}{}
\algtext*{EndForConstraint}

\algblockdefx[Constraint]{Constraint}{EndConstraint}{\textbf{Constraints:}}{}
\algtext*{EndConstraint}

\algblockdefx[Where]{Where}{EndWhere}{\textbf{Where:}}{}
\algtext*{EndWhere}

\newtcolorbox{myframe}[2][]{%
  enhanced,colback=white,colframe=black,coltitle=black,
  sharp corners, left=2pt, right=2pt, boxrule=0.5pt,
  fonttitle=\scshape\small,
  attach boxed title to top left={yshift=-0.3\baselineskip-0.4pt,xshift=3mm},
  boxed title style={tile,size=minimal,left=0.7mm,right=0.7mm,
    colback=white,before upper=\strut},
  title=#2,#1
}

\newenvironment{notation}%
{    \begin{myframe}{Notation}%
    \small%
}{    \end{myframe}%
}

\newcounter{statementcnt}
\newenvironment{statement}[1][]%
{    \stepcounter{statementcnt}%
    \begin{myframe}{Problem {\arabic{statementcnt}}\ifthenelse{\equal{#1}{}}{}{ (#1)}}%
    \small%
}{    \end{myframe}%
}

\newlength{\dhatheight}
\newcommand{\doublehat}[1]{%
    \settoheight{\dhatheight}{\ensuremath{\hat{#1}}}%
    \addtolength{\dhatheight}{-0.35ex}%
    \hat{\vphantom{\rule{1pt}{\dhatheight}}%
    \smash{\hat{#1}}}}

\newcommand{\QED}{\hfill \ensuremath{\Box}}

%% file: table-applications.tex
\begin{table*}[!htb]
    \newlength{\citationcolumnwidth}
    \setlength{\citationcolumnwidth}{\widthof{[00, 00, 00]}}
    
    \newcommand{\citbox}[1]{\parbox{\citationcolumnwidth}{\centering #1}}

        \newcommand{\oh}[1]{\ensuremath{\times\,#1}}
            \newcommand{\ohapprox}[1]{\oh{#1}}
        \newcommand{\nooh}{--}

        \newcommand{\WRBlackBox}{{\WRShort} (BB)}

    \centering
    \renewcommand{\arraystretch}{1.1}
    \resizebox{\textwidth}{!}{%
        \begin{tabular}{ccccccc}
        \hline
        \thead{distributed\\problem}
            & \thead{nominal\\solutions}
            & \thead{weight reduction\\problem}
            & \thead{$f_w$} 
            & \thead{$f_n$}
            & \thead{worst-case average\\comm. overhead}
            & \thead{worst-case average\\comp. overhead}
                \\
        \hline\hline
                \multicolumn{7}{c}{Derived Protocols}\\
                \hline\hline
        \makecell{Efficient Asynchronous\\State-Machine Replication} 
            & \citbox{\cite{honey-badger,beat-bft,all-you-need-is-dag,narwhal,bullshark}}
            & \makecell{{\WRShort} for RNG\\{\WQShort} for Broadcast}
            & $1/3$
            & $1/3$
            & \makecell{\ohapprox{1.33} for Broadcast \\ \oh{1.33} for RNG}
            & \makecell{\ohapprox{3.56} for Broadcast \\ \oh{1.33} for RNG}
                                    \\
        \hline
        \makecell{Structured Mempool}
            & \citbox{\cite{narwhal}}
            & {\WQShort} for Broadcast
            & $1/3$
            & $1/3$
            & \makecell{\ohapprox{1.33} for Broadcast}
            & \makecell{\ohapprox{3.56} for Broadcast}
                                    \\
        \hline
                                                                                                        \makecell{Validated Asynchronous\\Byzantine Agreement} 
            & \citbox{\cite{vaba-cachin,vaba-abraham}}
            & {\WRShort} for RNG
            & $1/3$
            & $1/3$
            & \oh{1.33} for RNG             & \oh{1.33} for RNG                                     \\

                                                                                                                                                                                                                                                                                        \hline
        \makecell{Consensus with Checkpoints}
            & \cite{pikachu}
            & {\WRShort} for signing
            & $1/3$ 
            & $1/3$ 
            & \oh{1.33} for signing
            & \oh{1.33} for signing
            \\
        \hline
        Linear BFT Consensus
            & \cite{yin2019hotstuff}
            & \multirowcell{2}{\WRBlackBox}
            & \multirowcell{2}{$1/4$}
            & \multirowcell{2}{$1/3$}
            & \multirowcell{2}{\oh{2.67}}
            & \multirowcell{2}{\oh{2.67}}
            \\ 
        Chain-Quality SSLE
            & \cite{boneh2020single}             \\
        \hline\hline

                \multicolumn{7}{c}{Useful Building Blocks} \\
        \hline\hline
                                                                                                                                                        \multirowcell{2}{Erasure-Coded\\Storage and Broadcast}
                        
                        & \multirowcell{2}{\citbox{\cite{ida-rabin-89,ida-cachin-tessaro-05,ida-hendricks-07,ida-nazirkhanova-21,ida-yang-22,rbc-erasures-high-threshold-trusted-setup}}}
            & {\WQShort}   & $1/3$ & $1/3$ & \ohapprox{1.33} & \ohapprox{3.56}
                        \\ \cline{3-7}
                    &             & {\WRBlackBox} & $1/4$ & $1/3$ & \nooh & \oh{3}
                        \\
        \hline
        \multirowcell{2}{Error-Corrected Broadcast}
            & \multirowcell{2}{\cite{dxr21}}
            & {\WQShort}   
            & $1/3$ 
            & $1/3$ 
            & \ohapprox{1.33}
            & \oh{7.11}
            \\ \cline{3-7}
                    &             & {\WRBlackBox}
            & $1/4$
            & $1/3$
            & \nooh
            & \oh{3}
            \\
        \hline
        Verifiable Secret Sharing & \cite{Sha79}
                        & \WRShort
            & $1/3$
            & $1/3$
            & \oh{1.33}
            & \oh{1.33}
                        \\
        \hline
        Common Coin & \cite{pre-shared-common-coins,random-oracles}
            & \multirowcell{4}{\WRShort}
            & \multirowcell{4}{$1/3$} 
            & \multirowcell{4}{$1/2$} 
            & \multirowcell{4}{\oh{1.33}}
            & \multirowcell{4}{\oh{1.33}}
            \\
        Blunt Threshold Signatures & \cite{ThresholdBLS,ThresholdSchnorr,rsa-threshold-signatures} \\
        Blunt Threshold Encryption & \cite{threshold-cryptosystems} \\
        Blunt Threshold FHE & \cite{ThFHE-jain,ThFHE-boneh} \\
        \hline
        Tight Secret Sharing
            & \multirowcell{4}{See sec. \ref{subsec:tight-secret-sharing}\\(this paper)}
            & \multirowcell{4}{\WRShort}
            & \multirowcell{4}{$1/2$} 
            & \multirowcell{4}{$1/2$} 
            & \multirowcell{4}{\oh{1.33} \\ (+$O(n^2)$ small messages)}
            & \multirowcell{4}{\oh{1.33}   }
                        \\
        Tight Threshold Signatures  \\
        Tight Threshold Encryption  \\
        Tight Threshold FHE  \\
        \hline
    \end{tabular}
    }
    \medskip
    \caption{Examples of suggested weighted distributed protocols with the upper bounds on communication and computation overhead compared to the nominal solutions with the same number of participants.
    See \Cref{sec:wr-applications,sec:wq-applications,sec:derived-applications} for details on how these numbers were obtained.
    In \Cref{sec:empirical-study}, we study real-world weight distributions and conclude that, in practice, the overhead should be much smaller.
    ``{\WRShort}'' and ``{\WQShort}'' refer to the weight reduction problems defined in \Cref{sec:problem-statement}. ``{\WRBlackBox}'' refers to the black-box transformation described in \Cref{subsec:black-box}.}
    \label{tab:applications}
\end{table*}

%% file: main.bbl
\begin{thebibliography}{10}

\bibitem{adkg}
Ittai Abraham, Philipp Jovanovic, Mary Maller, Sarah Meiklejohn, Gilad Stern,
  and Alin Tomescu.
\newblock Reaching consensus for asynchronous distributed key generation.
\newblock In {\em Proceedings of the 2021 ACM Symposium on Principles of
  Distributed Computing}, pages 363--373, Italy, virtual, 2021. ACM.

\bibitem{vaba-abraham}
Ittai Abraham, Dahlia Malkhi, and Alexander Spiegelman.
\newblock Asymptotically optimal validated asynchronous byzantine agreement.
\newblock In {\em Proceedings of the 2019 ACM Symposium on Principles of
  Distributed Computing}, pages 337--346, Toronto, 2019. ACM.

\bibitem{algorand_stake}
Algoexplorer.
\newblock Algorand stake distribution.
\newblock \url{https://algoexplorer.io/top-accounts}, 2023.
\newblock Accessed: 2023-03-28.

\bibitem{aptos_white}
Aptos.
\newblock White paper -- the aptos blockchain: Safe, scalable, and upgradeable
  web3 infrastructure.
\newblock Technical report, Aptos, 2022.
\newblock URL: \url{https://aptos.dev/aptos-white-paper/}.

\bibitem{aptos_stake}
Aptoscan.
\newblock Aptos stake distribution.
\newblock \url{https://aptoscan.com/validators?ps=100&p=}, 2023.
\newblock Accessed: 2023-03-28.

\bibitem{pikachu}
Sarah Azouvi and Marko Vukoli{\'c}.
\newblock Pikachu: Securing pos blockchains from long-range attacks by
  checkpointing into bitcoin pow using taproot.
\newblock In {\em Proceedings of the 2022 ACM Workshop on Developments in
  Consensus}, pages 53--65, Los Angeles, 2022. ACM.

\bibitem{BeiWei06}
Amos Beimel and Enav Weinreb.
\newblock Monotone circuits for monotone weighted threshold functions.
\newblock {\em Information Processing Letters}, 97(1):12--18, 2006.

\bibitem{ss-from-wiretap-channels}
Fabrice Benhamouda, Shai Halevi, and Lev Stambler.
\newblock Weighted secret sharing from wiretap channels.
\newblock Cryptology ePrint Archive, Paper 2022/1578, 2022.
\newblock \url{https://eprint.iacr.org/2022/1578}.
\newblock URL: \url{https://eprint.iacr.org/2022/1578}.

\bibitem{ThresholdBLS}
Alexandra Boldyreva.
\newblock Threshold signatures, multisignatures and blind signatures based on
  the gap-diffie-hellman-group signature scheme.
\newblock In Yvo~G. Desmedt, editor, {\em Public Key Cryptography --- PKC
  2003}, pages 31--46, Berlin, Heidelberg, 2002. Springer Berlin Heidelberg.

\bibitem{boneh2020single}
Dan Boneh, Saba Eskandarian, Lucjan Hanzlik, and Nicola Greco.
\newblock Single secret leader election.
\newblock In {\em Proceedings of the 2nd ACM Conference on Advances in
  Financial Technologies}, pages 12--24, New York, 2020. ACM.

\bibitem{ThFHE-boneh}
Dan Boneh, Rosario Gennaro, Steven Goldfeder, Aayush Jain, Sam Kim, Peter~MR
  Rasmussen, and Amit Sahai.
\newblock Threshold cryptosystems from threshold fully homomorphic encryption.
\newblock In {\em Advances in Cryptology--CRYPTO 2018: 38th Annual
  International Cryptology Conference, August 19--23, 2018, Proceedings, Part I
  38}, pages 565--596, Santa Barbara, CA, USA, 2018. Springer.

\bibitem{BraTou85}
Gabriel Bracha and Sam Toueg.
\newblock Asynchronous consensus and broadcast protocols.
\newblock {\em Journal of the ACM (JACM)}, 32(4):824--840, 1985.

\bibitem{textbook}
Christian Cachin, Rachid Guerraoui, and Lu{\'\i}s Rodrigues.
\newblock {\em Introduction to reliable and secure distributed programming}.
\newblock Springer Science \& Business Media, Berlin, Heidelberg, 2011.

\bibitem{avss-cachin-2002}
Christian Cachin, Klaus Kursawe, Anna Lysyanskaya, and Reto Strobl.
\newblock Asynchronous verifiable secret sharing and proactive cryptosystems.
\newblock In {\em Proceedings of the 9th ACM Conference on Computer and
  Communications Security}, pages 88--97, Washington, DC USA, 2002. ACM.

\bibitem{vaba-cachin}
Christian Cachin, Klaus Kursawe, Frank Petzold, and Victor Shoup.
\newblock Secure and efficient asynchronous broadcast protocols.
\newblock In {\em Advances in Cryptology—CRYPTO 2001: 21st Annual
  International Cryptology Conference, August 19--23, 2001 Proceedings}, pages
  524--541, Santa Barbara, California, USA, 2001. Springer.

\bibitem{random-oracles}
Christian Cachin, Klaus Kursawe, and Victor Shoup.
\newblock Random oracles in constantipole: practical asynchronous byzantine
  agreement using cryptography.
\newblock In {\em Proceedings of the nineteenth annual ACM symposium on
  Principles of distributed computing}, pages 123--132, Portland Oregon USA,
  2000. ACM.

\bibitem{ida-cachin-tessaro-05}
Christian Cachin and Stefano Tessaro.
\newblock Asynchronous verifiable information dispersal.
\newblock In {\em 24th IEEE Symposium on Reliable Distributed Systems
  (SRDS'05)}, pages 191--201, Orlando, Florida, USA, 2005. IEEE.

\bibitem{canetti-rabin}
Ran Canetti and Tal Rabin.
\newblock Fast asynchronous byzantine agreement with optimal resilience.
\newblock In {\em Proceedings of the twenty-fifth annual ACM symposium on
  Theory of computing}, pages 42--51, San Diego California USA, 1993. ACM.

\bibitem{pbft}
Miguel Castro and Barbara Liskov.
\newblock Practical byzantine fault tolerance.
\newblock In {\em Proceedings of the Third Symposium on Operating Systems
  Design and Implementation}, OSDI '99, page 173–186, USA, 1999. USENIX
  Association.

\bibitem{adaptive-ssle}
Dario Catalano, Dario Fiore, and Emanuele Giunta.
\newblock Adaptively secure single secret leader election from ddh.
\newblock In {\em Proceedings of the 2022 ACM Symposium on Principles of
  Distributed Computing}, pages 430--439, Salerno, Italy, 2022. ACM.

\bibitem{uc-ssle}
Dario Catalano, Dario Fiore, and Emanuele Giunta.
\newblock Efficient and universally composable single secret leader election
  from pairings.
\newblock In {\em Public-Key Cryptography--PKC 2023: 26th IACR International
  Conference on Practice and Theory of Public-Key Cryptography, May 7--10,
  2023, Proceedings, Part I}, pages 471--499, Atlanta, GA, USA, 2023. Springer.

\bibitem{ChaKia21}
Pyrros Chaidos and Aggelos Kiayias.
\newblock Mithril: Stake-based threshold multisignatures, 2021.

\bibitem{ColMarSav07}
Beno{\^\i}t Colson, Patrice Marcotte, and Gilles Savard.
\newblock An overview of bilevel optimization.
\newblock {\em Annals of operations research}, 153(1):235--256, 2007.

\bibitem{narwhal}
George Danezis, Lefteris Kokoris-Kogias, Alberto Sonnino, and Alexander
  Spiegelman.
\newblock Narwhal and tusk: a dag-based mempool and efficient bft consensus.
\newblock In {\em Proceedings of the Seventeenth European Conference on
  Computer Systems}, pages 34--50, Rennes, France, 2022. ACM.

\bibitem{das2023threshold}
Sourav Das, Philippe Camacho, Zhuolun Xiang, Javier Nieto, Benedikt Bunz, and
  Ling Ren.
\newblock Threshold signatures from inner product argument: Succinct, weighted,
  and multi-threshold, 2023.

\bibitem{das2024distributed}
Sourav Das, Benny Pinkas, Alin Tomescu, and Zhuolun Xiang.
\newblock Distributed randomness using weighted vrfs, 2024.

\bibitem{dxr21}
Sourav Das, Zhuolun Xiang, and Ling Ren.
\newblock Asynchronous data dissemination and its applications.
\newblock In {\em Proceedings of the 2021 ACM SIGSAC Conference on Computer and
  Communications Security}, pages 2705--2721, Virtual Event Republic of Korea,
  2021. ACM.

\bibitem{practical-adkg}
Sourav Das, Thomas Yurek, Zhuolun Xiang, Andrew Miller, Lefteris
  Kokoris-Kogias, and Ling Ren.
\newblock Practical asynchronous distributed key generation.
\newblock In {\em 2022 IEEE Symposium on Security and Privacy (SP)}, pages
  2518--2534, San Francisco, CA, USA, 2022. IEEE.

\bibitem{threshold-cryptosystems}
Yvo Desmedt.
\newblock Threshold cryptosystems.
\newblock In {\em Advances in Cryptology—AUSCRYPT'92: Workshop on the Theory
  and Application of Cryptographic Techniques Gold Coast, December 13--16, 1992
  Proceedings 3}, pages 1--14, Queensland, Australia, 1993. Springer.

\bibitem{beat-bft}
Sisi Duan, Michael~K Reiter, and Haibin Zhang.
\newblock Beat: Asynchronous bft made practical.
\newblock In {\em Proceedings of the 2018 ACM SIGSAC Conference on Computer and
  Communications Security}, pages 2028--2041, Toronto, Canada, 2018. ACM.

\bibitem{filecoin_stake}
Filfox.
\newblock Filecoin stake distribution.
\newblock \url{https://filfox.info/en/ranks/power}, 2023.
\newblock Accessed: 2023-03-28.

\bibitem{tezos_stake}
Fish.
\newblock Tezos stake distribution.
\newblock \url{https://tezos.fish/leaderboard/all}, 2023.
\newblock Accessed: 2023-03-28.

\bibitem{freitas2023homomorphic}
Luciano Freitas, Andrei Tonkikh, Adda-Akram Bendoukha, Sara Tucci-Piergiovanni,
  Renaud Sirdey, Oana Stan, and Petr Kuznetsov.
\newblock Homomorphic sortition--single secret leader election for pos
  blockchains, 2023.

\bibitem{GarJaiMukSinWanZha22}
Sanjam Garg, Abhishek Jain, Pratyay Mukherjee, Rohit Sinha, Mingyuan Wang, and
  Yinuo Zhang.
\newblock Cryptography with weights: Mpc, encryption and signatures, 2022.

\bibitem{rs-decoding-complexity}
Giuliano Garrammone.
\newblock On decoding complexity of reed-solomon codes on the packet erasure
  channel.
\newblock {\em IEEE Communications Letters}, 17(4):773--776, 2013.

\bibitem{weighted-voting}
David~K Gifford.
\newblock Weighted voting for replicated data.
\newblock In {\em Proceedings of the seventh ACM symposium on Operating systems
  principles}, pages 150--162, Pacific Grove California USA, 1979. ACM.

\bibitem{tezos_white}
L.M Goodman.
\newblock White paper -- tezos: a self-amending crypto-ledger.
\newblock Technical report, Tezos, 2014.
\newblock URL: \url{https://tezos.com/whitepaper.pdf}.

\bibitem{ida-hendricks-07}
James Hendricks, Gregory~R Ganger, and Michael~K Reiter.
\newblock Verifying distributed erasure-coded data.
\newblock In {\em Proceedings of the twenty-sixth annual ACM symposium on
  Principles of distributed computing}, pages 139--146, Portland Oregon USA,
  2007. ACM.

\bibitem{ItoSaiNis89}
Mitsuru Ito, Akira Saito, and Takao Nishizeki.
\newblock Secret sharing scheme realizing general access structure.
\newblock {\em Electronics and Communications in Japan (Part III: Fundamental
  Electronic Science)}, 72(9):56--64, 1989.

\bibitem{ThFHE-jain}
Aayush Jain, Peter~MR Rasmussen, and Amit Sahai.
\newblock Threshold fully homomorphic encryption, 2017.

\bibitem{rs-decoding-fft}
J{\o}rn Justesen.
\newblock On the complexity of decoding reed-solomon codes (corresp.).
\newblock {\em IEEE transactions on information theory}, 22(2):237--238, 1976.

\bibitem{all-you-need-is-dag}
Idit Keidar, Eleftherios Kokoris-Kogias, Oded Naor, and Alexander Spiegelman.
\newblock All you need is dag.
\newblock In {\em Proceedings of the 2021 ACM Symposium on Principles of
  Distributed Computing}, PODC'21, page 165–175, New York, NY, USA, 2021.
  Association for Computing Machinery.
\newblock \href {https://doi.org/10.1145/3465084.3467905}
  {\path{doi:10.1145/3465084.3467905}}.

\bibitem{KnapsackProblems04}
H.~Kellerer, U.~Pferschy, and D.~Pisinger.
\newblock {\em Knapsack Problems}.
\newblock Springer, Berlin, Germany, 2004.

\bibitem{filecoin_white}
Protocol Labs.
\newblock White paper -- filecoin: A decentralized storage network.
\newblock Technical report, Protocol Labs, 2017.
\newblock URL: \url{https://filecoin.io/filecoin.pdf}.

\bibitem{coding-theory-textbook}
Florence~Jessie MacWilliams and Neil James~Alexander Sloane.
\newblock {\em The theory of error-correcting codes}, volume~16.
\newblock Elsevier, Philadelphia, PA, USA, 1977.

\bibitem{byz-quorum-systems}
Dahlia Malkhi and Michael Reiter.
\newblock Byzantine quorum systems.
\newblock In {\em Proceedings of the Twenty-Ninth Annual ACM Symposium on
  Theory of Computing}, STOC '97, page 569–578, New York, NY, USA, 1997.
  Association for Computing Machinery.
\newblock \href {https://doi.org/10.1145/258533.258650}
  {\path{doi:10.1145/258533.258650}}.

\bibitem{algorand_white}
Silvio Micali.
\newblock {ALGORAND:} the efficient and democratic ledger, 2016.
\newblock URL: \url{http://arxiv.org/abs/1607.01341}, \href
  {https://arxiv.org/abs/1607.01341} {\path{arXiv:1607.01341}}.

\bibitem{honey-badger}
Andrew Miller, Yu~Xia, Kyle Croman, Elaine Shi, and Dawn Song.
\newblock The honey badger of bft protocols.
\newblock In {\em Proceedings of the 2016 ACM SIGSAC conference on computer and
  communications security}, pages 31--42, Vienna, Austria, 2016. ACM.

\bibitem{quorum-systems}
Moni Naor and Avishai Wool.
\newblock The load, capacity, and availability of quorum systems.
\newblock {\em SIAM Journal on Computing}, 27(2):423--447, 1998.

\bibitem{rbc-erasures-high-threshold-trusted-setup}
Kartik Nayak, Ling Ren, Elaine Shi, Nitin~H. Vaidya, and Zhuolun Xiang.
\newblock {Improved Extension Protocols for Byzantine Broadcast and Agreement}.
\newblock In Hagit Attiya, editor, {\em 34th International Symposium on
  Distributed Computing (DISC 2020)}, volume 179 of {\em Leibniz International
  Proceedings in Informatics (LIPIcs)}, pages 28:1--28:17, Dagstuhl, Germany,
  2020. Schloss Dagstuhl -- Leibniz-Zentrum f{\"u}r Informatik.
\newblock URL:
  \url{https://drops.dagstuhl.de/entities/document/10.4230/LIPIcs.DISC.2020.28},
  \href {https://doi.org/10.4230/LIPIcs.DISC.2020.28}
  {\path{doi:10.4230/LIPIcs.DISC.2020.28}}.

\bibitem{ida-nazirkhanova-21}
Kamilla Nazirkhanova, Joachim Neu, and David Tse.
\newblock Information dispersal with provable retrievability for rollups, 2021.

\bibitem{ohta1999multi}
Kazuo Ohta and Tatsuaki Okamoto.
\newblock Multi-signature schemes secure against active insider attacks.
\newblock {\em IEICE Transactions on Fundamentals of Electronics,
  Communications and Computer Sciences}, 82(1):21--31, 1999.

\bibitem{pre-shared-common-coins}
Michael~O Rabin.
\newblock Randomized byzantine generals.
\newblock In {\em 24th annual symposium on foundations of computer science
  (sfcs 1983)}, pages 403--409, Tucson, Arizona, USA, 1983. IEEE.

\bibitem{ida-rabin-89}
Michael~O Rabin.
\newblock Efficient dispersal of information for security, load balancing, and
  fault tolerance.
\newblock {\em Journal of the ACM (JACM)}, 36(2):335--348, 1989.

\bibitem{sok-randomness-beacons}
Mayank Raikwar and Danilo Gligoroski.
\newblock Sok: Decentralized randomness beacon protocols.
\newblock In {\em Information Security and Privacy: 27th Australasian
  Conference, ACISP 2022, November 28--30, 2022, Proceedings}, pages 420--446,
  Wollongong, NSW, Australia, 2022. Springer.

\bibitem{Sha79}
Adi Shamir.
\newblock How to share a secret.
\newblock {\em Communications of the ACM}, 22(11):612--613, 1979.

\bibitem{rsa-threshold-signatures}
Victor Shoup.
\newblock Practical threshold signatures.
\newblock In {\em Advances in Cryptology—EUROCRYPT 2000: International
  Conference on the Theory and Application of Cryptographic Techniques, May
  14--18, 2000 Proceedings 19}, pages 207--220, Bruges, Belgium, 2000.
  Springer.

\bibitem{bullshark}
Alexander Spiegelman, Neil Giridharan, Alberto Sonnino, and Lefteris
  Kokoris-Kogias.
\newblock Bullshark: Dag bft protocols made practical.
\newblock In {\em Proceedings of the 2022 ACM SIGSAC Conference on Computer and
  Communications Security}, CCS '22, page 2705–2718, New York, NY, USA, 2022.
  Association for Computing Machinery.
\newblock \href {https://doi.org/10.1145/3548606.3559361}
  {\path{doi:10.1145/3548606.3559361}}.

\bibitem{ThresholdSchnorr}
Douglas~R. Stinson and Reto Strobl.
\newblock Provably secure distributed schnorr signatures and a (t, n) threshold
  scheme for implicit certificates.
\newblock In {\em Proceedings of the 6th Australasian Conference on Information
  Security and Privacy}, ACISP '01, page 417–434, Berlin, Heidelberg, 2001.
  Springer-Verlag.

\bibitem{ida-yang-22}
Lei Yang, Seo~Jin Park, Mohammad Alizadeh, Sreeram Kannan, and David Tse.
\newblock Dispersedledger: High-throughput byzantine consensus on variable
  bandwidth networks.
\newblock In {\em 19th USENIX Symposium on Networked Systems Design and
  Implementation (NSDI 22)}, pages 493--512, Renton, WA, USA, 2022. USENIX.

\bibitem{yin2019hotstuff}
Maofan Yin, Dahlia Malkhi, Michael~K Reiter, Guy~Golan Gueta, and Ittai
  Abraham.
\newblock Hotstuff: Bft consensus with linearity and responsiveness.
\newblock In {\em Proceedings of the 2019 ACM Symposium on Principles of
  Distributed Computing}, pages 347--356, Toronto ON Canada, 2019. ACM.

\end{thebibliography}
